\documentclass[a4paper,oneside,11pt]{article}
\usepackage{import}
\usepackage{fullpage}
\usepackage{fancyhdr}
\AtBeginEnvironment{pmatrix}{\setlength{\arraycolsep}{8pt}}
\usepackage{setspace}

\usepackage{amsmath,etoolbox}
\usepackage{amssymb}
\usepackage[utf8x]{inputenc} 
\usepackage{accents}
\usepackage{bbm}
\usepackage{mathtools}
\usepackage{amsthm}
\usepackage[english]{babel}
\usepackage{letltxmacro}
\usepackage{xparse}
\usepackage[dvipsnames,svgnames,x11names]{xcolor}
\usepackage{graphicx}
\usepackage{multirow}
\usepackage{multicol}
\usepackage{nicefrac}

\usepackage{tikz}
\usetikzlibrary{decorations.pathreplacing}
\usetikzlibrary{decorations.markings}
\usepackage{wrapfig}
\usepackage{caption}
\usepackage{subcaption}
\usepackage[percent]{overpic}
\usepackage{tikz-cd}

\usepackage{ebgaramond}
\usepackage{euscript}
\usepackage{mathrsfs}
\usepackage{bm}
\usepackage{commath}
\usepackage{accents}
\usepackage[bbgreekl]{mathbbol}
\usepackage[euler]{textgreek}
\usepackage[scr=boondoxo]{mathalfa}
\DeclareSymbolFontAlphabet{\mathbbm}{bbold}
\DeclareSymbolFontAlphabet{\mathbb}{AMSb}%
\usepackage{mleftright}
\usepackage{dirtytalk}
\usepackage{shuffle}
\usepackage{scalerel}
\usepackage{tensor}
\usepackage{soul}
\setul{0pt}{0.4pt}
\usepackage[T1]{fontenc}

\usepackage[colorlinks=true]{hyperref}
\definecolor{imperialBlue}{RGB}{0, 62, 116}
\definecolor{imperialBrick}{RGB}{165,25,0}
\definecolor{imperialProcess}{RGB}{0,133,202}
\definecolor{imperialGreen}{RGB}{2,137,59}
\definecolor{imperialRed}{RGB}{221,37,1}
\definecolor{imperialOrange}{RGB}{210,64,0}
\definecolor{imperialBlue2}{RGB}{0,110,175}
\definecolor{imperialTangerine}{RGB}{236,115,0}
\definecolor{imperialPurple}{RGB}{101,48,152}
\definecolor{imperialLime}{RGB}{196,214,0}
\definecolor{imperialKermit}{RGB}{102,164,10}
\hypersetup{
    colorlinks,
    citecolor=imperialGreen,
    filecolor=imperialBlue,
    linkcolor=imperialBlue,
    urlcolor=imperialProcess
}

\usepackage[nameinlink, capitalise]{cleveref}
\usepackage{mdframed}
\newtheorem{theorem}{Theorem}[section]
\newtheorem{corollary}[theorem]{Corollary}
\newtheorem{lemma}[theorem]{Lemma}
\newtheorem{proposition}[theorem]{Proposition}
\newtheorem{remark}[theorem]{Remark}
\theoremstyle{definition}
\newtheorem{definition}[theorem]{Definition}

\newtheorem{example}[theorem]{Example}

\usepackage{algorithm}
\usepackage[noend]{algpseudocode}

\usepackage{listings}
\usepackage{enumitem}
\usepackage[nottoc,notlot,notlof]{tocbibind}

\newcommand{\C}{\mathbb{C}}

\newcommand{\Ex}{\mathbb{E}}

\newcommand{\N}{\mathbb{N}}

\renewcommand{\P}{\mathbb{P}}

\newcommand{\R}{\mathbb{R}}

\renewcommand{\AA}{\mathcal{A}}

\newcommand{\CC}{\mathcal{C}}
\newcommand{\DD}{\mathcal{D}}

\newcommand{\HH}{\mathcal{H}}

\newcommand{\MM}{\mathcal{M}}
\newcommand{\NN}{\mathcal{N}}
\newcommand{\OO}{\mathcal{O}}
\newcommand{\PP}{\mathcal{P}}

\renewcommand{\SS}{\mathcal{S}}

\newcommand{\WW}{\mathcal{W}}
\newcommand{\XX}{\mathcal{X}}

\DeclarePairedDelimiterX{\normop}[1]{\lVert}{\rVert_{\mathrm{op}}}{#1}
\DeclarePairedDelimiterX{\normopp}[1]{\lVert}{\rVert_{\mathrm{op}}^p}{#1}
\DeclarePairedDelimiterX{\normhs}[1]{\lVert}{\rVert_{\mathrm{HS}}}{#1}

\newcommand{\tr}[1]{\mathrm{tr} \big(#1\big)}
\newcommand{\Tr}[1]{\mathrm{Tr} \big(#1\big)}
\newcommand{\pdg}[2]{\langle\gamma_{#1,#2}\rangle_{\mu_V^\infty}}
\newcommand{\pdgn}[2]{\langle\gamma_{#1,#2}\rangle_{\mu_V^N}}

\newcommand{\wt}[1]{\mathrm{wt} (#1)}

\newcommand{\id}{\mathbbm{1}}

\renewcommand{\del}{\partial}

\makeatletter
\def\upintkern@{\mkern-7mu\mathchoice{\mkern-3.5mu}{}{}{}}
\def\upintdots@{\mathchoice{\mkern-4mu\@cdots\mkern-4mu}%
	{{\cdotp}\mkern1.5mu{\cdotp}\mkern1.5mu{\cdotp}}%
	{{\cdotp}\mkern1mu{\cdotp}\mkern1mu{\cdotp}}%
	{{\cdotp}\mkern1mu{\cdotp}\mkern1mu{\cdotp}}}

\newcommand{\UpMultiIntegral}[1]{%
	\edef\ints@c{\noexpand\upintop
		\ifnum#1=\z@\noexpand\upintdots@\else\noexpand\upintkern@\fi
		\ifnum#1>\tw@\noexpand\upintop\noexpand\upintkern@\fi
		\ifnum#1>\thr@@\noexpand\upintop\noexpand\upintkern@\fi
		\noexpand\upintop
		\noexpand\ilimits@
	}%
	\futurelet\@let@token\ints@a
}

\makeatletter
\newcommand{\di}{\mathrm{d}}
\renewcommand*{\dif}%
    {\@ifnextchar^{\DIfF}{\di}} 
\def\DIfF^#1{%
    \mathop{\mathrm{\mathstrut d}}%
        \nolimits^{#1}\gobblespace}
\def\gobblespace{%
        \futurelet\diffarg\opspace}
\def\opspace{%
        \let\DiffSpace\!%
        \ifx\diffarg(%
            \let\DiffSpace\relax
        \else
            \ifx\diffarg[%
               \let\DiffSpace\relax
            \else
               \ifx\diffarg\{%
                   \let\DiffSpace\relax
               \fi\fi\fi\DiffSpace}
\makeatother

\DeclareFontFamily{OMX}{mdbch}{}
\DeclareFontShape{OMX}{mdbch}{m}{n}{ <->s * [0.8]  mdbchr7v }{}
\DeclareFontShape{OMX}{mdbch}{b}{n}{ <->s * [0.8]  mdbchb7v }{}
\DeclareFontShape{OMX}{mdbch}{bx}{n}{<->ssub * mdbch/b/n}{}

\DeclareSymbolFont{uplargesymbols}{OMX}{mdbch}{m}{n}
\SetSymbolFont{uplargesymbols}{bold}{OMX}{mdbch}{b}{n}
\DeclareMathSymbol{\upintop}{\mathop}{uplargesymbols}{82}
\DeclareMathSymbol{\upointop}{\mathop}{uplargesymbols}{"48}

\DeclareFontEncoding{MDB}{}{}
\DeclareFontFamily{MDB}{mdbch}{}
\DeclareFontShape{MDB}{mdbch}{m}{n}{ <->s * [0.8]  mdbchrmb }{}
\DeclareFontShape{MDB}{mdbch}{b}{n}{ <->s * [0.8]  mdbchbmb }{}
\DeclareFontShape{MDB}{mdbch}{bx}{n}{<->ssub * mdbch/b/n}{}
\DeclareFontSubstitution{MDB}{cmr}{m}{n}
\DeclareSymbolFont{mathdesignB}{MDB}{mdbch}{m}{n}%
\SetSymbolFont{mathdesignB}{bold}{MDB}{mdbch}{b}{n}%
\DeclareMathSymbol{\upintclockwise}{\mathop}{mathdesignB}{128}
\DeclareMathSymbol{\upointclockwise}{\mathop}{mathdesignB}{130}
\DeclareMathSymbol{\upointctrclockwise}{\mathop}{mathdesignB}{132}
\DeclareMathSymbol{\upoiint}{\mathop}{mathdesignB}{134}
\DeclareMathSymbol{\upoiiint}{\mathop}{mathdesignB}{136}

\makeatletter
\newcommand{\upint}{\DOTSI\upintop\ilimits@}
\newcommand{\upoint}{\DOTSI\upointop\ilimits@}
\makeatother

\renewcommand{\int}{\upint}

\usepackage{stackengine}

\newcommand{\Cd}{\C\langle X_1,\dots,X_d\rangle}

\usepackage[left=1in, right=1in, top=1in, bottom=1in]{geometry}

\usepackage[textsize=tiny]{todonotes}

\usepackage{physics}
\usepackage{tikz}
\usepackage{quantikz}

\renewcommand{\tr}[1]{\mathrm{tr} \big(#1\big)}
\renewcommand{\Tr}[1]{\mathrm{Tr} \big(#1\big)}

\usepackage{titling}
\usepackage{authblk}

\title{\textsc{Quantum Path Signatures}}

\author{
Samuel Crew\textsuperscript{1,2,*}, 
Cristopher Salvi\textsuperscript{1,*}, 
William F. Turner\textsuperscript{1,*}, 
Thomas Cass\textsuperscript{1}, 
Antoine Jacquier\textsuperscript{1}
}

\affil{\small
\textsuperscript{1}Department of Mathematics, Imperial College London, UK.\\
\textsuperscript{2}Department of Physics, National Tsing Hua University, Taiwan.
}

\date{\today}

\usepackage[style=alphabetic, giveninits = true, maxbibnames = 99]{biblatex}
\usepackage{csquotes}
\begin{document}

\maketitle

\footnotetext[1]{These authors contributed equally to this work.}

\vspace{-1em} 
\begin{abstract}
    We elucidate physical aspects of path signatures by formulating randomised path developments within the framework of matrix models in quantum field theory. Using tools from physics, we introduce a new family of randomised path developments and derive corresponding loop equations. We then interpret unitary randomised path developments as time evolution operators on a Hilbert space of qubits. This leads to a definition of a quantum path signature feature map and associated quantum signature kernel through a quantum circuit construction. In the case of the Gaussian matrix model, we study a random ensemble of Pauli strings and formulate a quantum algorithm to compute such kernel.
\end{abstract}

\setcounter{tocdepth}{2}
\vspace{-1em} 
\tableofcontents

\section{Introduction}

On sufficiently fine time scales, many types of sequential data, such as audio, video, time series, or text, can naturally be modelled as paths $\gamma : [0,1] \to \mathbb{R}^d$. The idea that such paths can be characterised (up to reparameterisation) by their \emph{path signature} was first introduced by Chen~\cite{chen1957integration} for smooth paths, and later extended to more irregular paths in the context of \emph{rough path theory}~\cite{lyons1998differential, hambly2010uniqueness, boedihardjo2016signature}.  In this work, we elucidate a connection between \emph{quantum field theory} (QFT) and the theory of path signatures. Adopting the language of QFT, refined later in Section~\ref{sec:background}, we treat a path $\gamma$ as a Wilson line and interpret its signature as arising from parallel transport in gauge theory. This perspective allows us to draw on tools from high-energy theoretical physics to study a class of differential equations driven by random vector fields sampled from so-called matrix model distributions, which develop the path into a unitary group. An important instance of this framework is the Gaussian unitary development, whose physical interpretation naturally leads to a representation of the control system as a quantum circuit acting on a Hilbert space of qubits. This construction, in turn, gives rise to a quantum algorithm whose output approximates the associated quantum signature kernels.

The path signature $\mathcal{S}(\gamma)$ is formally defined as the solution to the tensor differential equation \( \dif y = y \otimes \dif\gamma \) in the free tensor algebra \( T((\mathbb{R}^d)) := \prod_{n=0}^\infty (\mathbb{R}^d)^{\otimes n} \), and thereby can be seen as a non-commutative analogue of the exponential function. By applying Picard iteration, one recovers the familiar formulation of the signature as the sequence of iterated integrals $\left( \int_{0<t_1<\cdots<t_n<1} \dif\gamma_{t_1} \otimes \cdots \otimes \dif\gamma_{t_n} \right)_{n \in \mathbb{N}}$. It is a classical result that linear functionals on \( T((\mathbb{R}^d)) \), when restricted to the image of the signature map, form a unital algebra that separates points. Consequently, the classical Stone--Weierstrass theorem ensures that linear functionals on signatures are dense in the space of continuous real-valued functions defined on compact subsets of unparameterised paths~\cite{CT_topologies_a}. This makes the signature a powerful representation, enabling the approximation of arbitrary path-dependent functionals using linear models~\cite{salvi2023structure}. Indeed, signature methods have seen a rapid rise in popularity in recent years, finding applications across a broad range of domains in data science, including deep learning for sequence modelling~\cite{kidger2019deep, morrill2021neural, salvi2022neural, hoglund2023neural, cirone2024theoretical, issa2024non, barancikova2024sigdiffusions, cirone2025parallelflow}, quantitative finance~\cite{arribas2020sigsdes, salvi2021higher, horvath2023optimal, pannier2024path, cirone2025rough}, cybersecurity~\cite{cochrane2021sk}, and computational neuroscience~\cite{holberg2024exact}, among many others. The interested reader is referred to \cite{cass_salvi_notes} for a detailed and pedagogical account on the subject and to \cite{fermanian2023new} for an overview of recent applications. 

In practice, the paths encountered in data-driven applications are typically constructed via piecewise linear interpolation of discrete time series. By combining Chen's identity with the fact that the signature of a linear segment is given by the tensor exponential of its increment, one obtains a concise formula for the signature of a piecewise linear path \( \gamma = \gamma^1 * \cdots * \gamma^L \) as $\mathcal S(\gamma) = \exp(\Delta \gamma^1) \otimes \cdots \otimes \exp(\Delta \gamma^L)$. This expression underlies the implementation of signature computations in standard Python libraries such as \texttt{esig}~\cite{esig}, \texttt{iisignature}~\cite{reizenstein2018iisignature}, and \texttt{signatory}~\cite{signatory}. It enables an efficient and exact evaluation of the signature, with time complexity \(\OO(Ld^n)\), where \(n \in \mathbb{N}\) denotes the truncation level. However, due to the exponential growth in the dimension \(d\), the method becomes computationally prohibitive as \(n\) increases. A commonly-taken approach solution to this dimensionality issue making use of signature kernel methods has become popular~\cite{ lemercier2021distribution, salvi2021rough, lemercier2021siggpde, manten2024signature, shmelev2024sparse}. These allow for efficient computation of kernels of the form \(\langle \mathcal{S}(\gamma), \mathcal{S}(\sigma) \rangle,\) for various choices of inner product \(\langle \cdot, \cdot \rangle\) on \(T((\mathbb R^d))\), without the need for computing \(\mathcal{S}(\gamma)\) directly, for example leveraging the recent result that the signature kernel satisfies a Goursat PDE \cite{SigPDE}.

Another promising recent alternative to these ideas is a proposal to use the development of \(\gamma\) into a matrix Lie group~\cite{lou2022path}. In this formulation, the central object becomes not the signature but related representations obtained by the evolution 
\begin{equation}\label{eqn:first_pd}
    \dif U(\gamma; M)_t = U(\gamma; M)_t \cdot M(\dif\gamma_t) = \sum_{j=d}^d U(\gamma; M)_t A_j \dif\gamma^j, \quad U(\gamma;M)_0 = I_N,
\end{equation}
where \(U(\gamma; M)\) belongs to a matrix Lie group \(G_N\), \(M\) is a map in \(\mathrm{Hom}(\mathbb R^d, \mathfrak{g}_N)\) such that $M(v) = \sum_{j=1}^d A_jv^j$ for $A_1,\cdots,A_d \in \mathfrak{g}_N$, with $\mathfrak{g}_N$ denoting the associated Lie algebra, the latter being isomorphic to the space of left-invariant vector fields on \(G_N\), and in which \(\cdot\) denotes matrix multiplication. In this work we consider the physical and geometrical interpretation of the development of $\gamma$ as a parallel transport operator on a $G_N$ bundle. 
The solution to~\eqref{eqn:first_pd} can be expressed in terms of the signature as
\begin{equation}
    U(\gamma; M)_t =\sum_{\bm{w}\in\WW_d}i^{|\bm{w}|}A_{\bm{w}}\SS_{0,t}^{\bm{w}}(\gamma),
\end{equation}
where the sum is taken over all words in~$d$ letters. An important example of this framework is when \(G_N = U(N)\) is the unitary group, whose corresponding Lie algebra $\mathfrak{g}_N = \mathfrak{u}(N)$ is the set of anti-Hermitian matrices. We note that the critical issue of dimensionality outlined above is transferred, rather than resolved, from the selection of a truncation parameter in the signature to the choice of a Lie group of sufficiently high dimension for the method to be effective. A promising approach recently proposed in the literature~\cite{CT_free, SKlimit} consists of randomising the choice of vector fields $M = (A_1,\ldots,A_d)\sim\mu_V^N$ over a suitable class of (Borel) probability measures $\mu_V^N$ on the Lie algebra and study the limiting behaviour of (inner products of such) path developments as $N$ tends to infinity. 

The physical perspective adopted in this work interprets path development as parallel transport in gauge theory. This viewpoint leads to an interpretation of vector field randomisation as a zero-dimensional quantum field theory, where integration over gauge fields is realised via a matrix-valued path integral. Conversely, we introduce path developments as a novel class of observables in matrix model theory, interpolating between zero-dimensional quantum field theory and Yang–Mills theory. More precisely, we will be interested in characterising the large-$N$ behaviour of the law of the random variables
\begin{equation}
    \langle U(\gamma;M), U(\sigma; M) \rangle_{\text{HS}}:= \lim_{N\to \infty}\frac{1}{N}\Tr{U(\gamma;M) U(\sigma; \cdot)^\dag} \quad \text{under } M \sim \mu_V^N,
\end{equation}
where, importing tools from theoretical physics, we consider $\mu_V^N$ a matrix model measure on $\mathrm{u}(N)$ given by~\cite{comb_matrix_models,matrix_models}
\begin{equation}
    \dif \mu_V^N (M) \propto \exp\left(-N \tr{V(M)}\right)\dif M,
\end{equation}
where $\tr{\cdot}$ denotes the normalised trace defined by $\tr{\cdot} := \frac{1}{N}\Tr{\cdot}$ and $V(A)$ denotes a Gaussian-perturbed polynomial potential 
\begin{equation}
    V(A) = V(A_1,\cdots,A_d) = \frac{1}{2} \sum_{i=1}^d A_i^2 + W,
\end{equation}
where $W$ is a polynomial satisfying some suitable convexity conditions. By combining techniques from non-commutative probability theory~\cite{Guionnet_SD} and rough path theory \cite{cass_salvi_notes}, we can define the limit
\begin{equation}
    k^{\mu_V^\infty}(\gamma,\sigma) := \lim_{N \to \infty} \mathbb E_{M \sim \mu_V^N}\left[\tr{U(\gamma;M) U(\sigma;M)^{\dag}}\right],
\end{equation}
and moreover understand $k^{\mu_V^\infty}(\gamma,\sigma)$ in terms of the signature of the path $y = \gamma * \overleftarrow{\sigma}$:
\begin{align}
    k^{\mu_V^\infty}(\gamma,\sigma) &= \sum_{\bm{w}\in\WW_d}i^{|\bm{w}|} \lim_{N\to\infty}\mathbb{E}_{\mu_V^N}\big[\tr{A_{\bm{w}}}\big]\SS^{\bm{w}}(y) \\
    &:= \sum_{\bm{w}\in\WW_d}i^{|\bm{w}|} \tau_V(A_{\bm{w}})\SS^{\bm{w}}(y),
\end{align}
where $\tau_V(A_{\bm{w}})$ is a non-commutative law satisfying the Schwinger-Dyson equation (\Cref{thm: S-D_uniqueness}). We use this structure to prove that $k^{\mu_V^\infty}(\gamma,\sigma) = k^y(1,1)$, where $k^y : [0,1]^2 \to \mathbb R$ satisfies an integro-differential equation of the form (\Cref{thm: integro_diff_eq})
\begin{equation}\label{eqn:pde}
\pdg{s}{t}=1-\int\limits_{s\leq u\leq v\leq t}\pdg{s}{u}\pdg{u}{v}\langle\dif\gamma_u,\dif\gamma_v\rangle - \sum_{k=1}^d\int_s^t\DD_V^k\pdg{s}{u}\dif\gamma_u^k,
\end{equation}
with the boundary condition $\pdg{s}{s}=1$ for all $s\in [0,T]$, and where the pathwise derivative $\DD_V^k$ will be rigorously defined in \Cref{sec:unitarydevelopments}.

This result generalises the findings in \cite{CT_free}, where the special case of quadratic potential $V=\tfrac{1}{2}\sum_{i=1}^dX_i^2$ was considered along with variants\footnote{The work \cite{CT_free} shows, in addition, stability of the limit beyond the setting of i.i.d. Gaussian entries.} leading to the same limiting PDE, for which $\mu_V^N$ corresponds to the classical Gaussian unitary ensemble (GUE). 
We show in addition that the integro-differential equation~\eqref{eqn:pde} has a unique solution when the solution space is restricted to (infinite) linear functionals on the signature.

In Section~\ref{sec:quantumsignature} we turn our attention to the GUE ensemble and piecewise linear paths $\gamma: [0,T] \to \mathbb{R}^d$ given as $L$ specified increment vectors \( \{\Delta_l^1, \Delta_l^2, \ldots, \Delta_l^d\}_{l=1}^L \). We consider the path development as a unitary operator on $\mathcal{H} = (\mathbb{C}^2)^{\otimes n}$, a Hilbert space of $n$ qubits. A linear basis of generators of unitary operators on qubits is given by Pauli strings: length $n$ tensor products of the Pauli matrices $P = \{\sigma_I,\sigma_X,\sigma_Y,\sigma_Z\}$. We define a random ensemble of linear combinations of Pauli strings
\begin{equation}
    A_{\nu} = \sum_{\mathbf{w} \in P} \alpha_{\nu}^{\mathbf{w}}\sigma_{\mathbf{w}},
    \quad \nu=1,\ldots, d,
\end{equation}
where the coefficients are randomly distributed as defined precisely in \cref{def:SPS}.

We first demonstrate that the GUE randomised path development may be approximated by a development driven by this sparse ensemble of random Pauli strings (\cref{thm:sparse_approximation}). 
We then consider a Trotterisation of this path development into~$K$ subdivisions and thereby define a random quantum circuit approximating the unitary development. The circuit depends on the parameters: $K$ Trotter subdivisions, $m$ densities of Pauli ensemble and~$n$ qubits, and is given by
\begin{equation}
    U^Q_{\gamma}(\alpha(m),n,K) = \prod_{l=1}^L \left[\prod_{\nu=1}^d \prod_{i=1}^m P_{\mathbf{w}_{i,\nu}}(\Delta_l^{\nu} \alpha_{\nu}^{\mathbf{w}_i}/K)\right]^K,
\end{equation}
where $P_{\mathbf{w}}(\theta)$ is a Pauli rotation of angle $\theta$ associated to the Pauli string $\sigma_{\mathbf{w}}$. Finally, we use this circuit to propose an efficient quantum algorithm (\cref{thm:quantumalg}) in the one clean qubit model of quantum computation to approximate the GUE path development and the associated kernel.

\paragraph{Outline and summary.} This work lies at the intersection of matrix models in quantum field theory (QFT) and the theory of path signatures, with a particular focus on their applications in machine learning and data science. Our aim is to make this work accessible to researchers from both communities. To that end, we provide a concise introduction to the relevant background and establish our notation in \Cref{sec:background}. In \Cref{sec:unitarydevelopments}, we apply classical tools from non-commutative probability theory to derive the integro-differential equation (\ref{eqn:pde}) satisfied by $\pdg{s}{t}$ under small perturbations of the Gaussian potential defined by matrix models. Finally, in \Cref{sec:quantumsignature}, we introduce quantum analogues of the classical path signature and signature kernel. The main results of this section demonstrate that the path development driven by the random Pauli ensemble approximates well the unitary path development and that the associated quantum algorithm is efficient.

\paragraph{Acknowledgements.}
SC, AJ and CS are supported by the Innovate UK "Quantum Machine Learning for Financial Data Streams" (Project ref 10073285), in collaboration with Standard Chartered and Rigetti Computing. SC, AJ and CS are grateful to Marco Paini (Rigetti),  Ernesto Palidda (Rigetti) and Mattia Fiorentini (Rigetti) for many interesting discussions. 
AJ is also supported by  the EPSRC Grant EP/W032643/1.
TC is  supported by the EPSRC Programme Grant EP/S026347/1.
SC is supported by National Science and Technology Council of Taiwan (NSTC 113-2112-M-007-019) and is grateful to Masazumi Honda (RIKEN) and Heng-Yu Chen (NTU) for many interesting discussions and RIKEN for hospitality whilst part of this work was completed. WFT has been supported by the EPSRC Centre for Doctoral Training in Mathematics of Random Systems: Analysis, Modelling and Simulation (EP/S023925/1). SC, WFT and CS are grateful to Pochung Cheng and Po-Yao Chang for their organisation of the workshop `Random matrices and quantum machine learning' at National Tsing Hua University, Taiwan where part of this work was finalised. 
For the purpose of open access, the authors have applied for a creative commons attribution (CC BY) licence to any author accepted manuscript version arising.

\section{Background}\label{sec:background}
In this section, we provide a brief introduction to quantum field theory (QFT) in zero dimensions and path signatures and set the notation for the rest of the work. We also introduce a novel geometric-physical interpretation of path signatures and signature kernels as matrix model Wilson line correlation functions.

\subsection{Quantum field theory}
The physical context of our work is that of zero-dimensional QFT in the form of Hermitian matrix models. A thorough introduction to zero-dimensional QFT aimed at a mathematical audience may be found in \cite{hori2003mirror} and we refer the reader to \cite{eynard2015random} for a comprehensive review of matrix model theory. 
The theory of matrix models is a classical subject in theoretical Physics~\cite{t1993planar} with deep connections to enumerative geometry via topological recursion \cite{eynard2007invariants} and quantum gravity in two dimensions \cite{gross1990nonperturbative}. We provide a cursory overview of the subject and focus on reviewing the basic aspects of loop operators that the present work generalises.

We consider matrix-valued fields fluctuating on a zero-dimensional spacetime. Let us consider the unitary matrix Lie group
\begin{equation}
    U_N=\{U\in\MM_N(\C):U^\dag U = I_N\},
\end{equation}
where $\mathcal{M}_{N}(\mathbb{C})$ denotes $N\times N$ complex matrices, and the corresponding Lie algebra
\begin{equation}
    \mathfrak{u}_N=\{A\in\MM_N(\C) : A^\dag = -A\}.
\end{equation}
We now consider a connection (or gauge field in the Physics terminology); a Lie algebra-valued one-form which, in coordinates, may be expressed as an $N\times N$ matrix $A \in \mathfrak{u}_N$. Global gauge transformations act on these matrices by the adjoint action $g \cdot A = g^{-1} A g$ with $g \in U_N$. In QFT, the Feynman path integral averages over fluctuating field configurations weighted by a gauge-invariant action. In the present context we integrate over anti-Hermitian $N \times N$ matrices and the measure we choose is a so-called matrix model measure 
\begin{equation*}
    \dif \mu_V (A) = \exp\left(-N \tr{V(A)}\right)\dif A.
\end{equation*}
We study interacting polynomial potentials of the form
\begin{equation*}
    V(A,\{g_k\}) = \frac{1}{2} A^2 + \sum_{k=3}^m \frac{g_k}{k!} A^k,
\end{equation*}
depending on a set of coupling constants $\{g_k\}_{k \ge 3}$ and $A \in\mathfrak{u}_N$. We will provide a more mathematically precise exposition of such matrix model measures in the context of non-commutative probability theory in \cref{sec:unitarydevelopments}. Finally, we should specify the set of (gauge-invariant) observables we are interested in computing. We write expectation values as
\begin{equation}\label{eq:ExpectationValue}
    \left\langle \tr{f(A)} \right\rangle := \frac{1}{Z}\int \mathrm{d} \mu_{V}(A)\, \tr{f(A)},
\end{equation}
for a matrix function $f(A)$ and where $Z:=\langle 1 \rangle$ denotes the partition function. 
Solving the matrix model corresponds to determining the moment generating function. In the matrix model context, this object is the multi-trace correlation function\footnote{The right hand side of this expression is shorthand for the formal generating function of the single trace correlators $\langle \tr{A^{l_1}}\ldots \tr{A^{l_n}} \rangle_c$.}
\begin{equation}
    W_n\left(x_1,\ldots,x_n;\{g_k\}\right) := \left\langle \tr{\tfrac{1}{x_1 I - A}}\cdots \tr{\tfrac{1}{x_n I - A}}\right\rangle_{\mu_V,c}
\end{equation}
where the subscript $c$ denotes the connected part of the correlator. Correlation function with $n$ insertions enjoy a formal asymptotic expansion, the so-called genus expansion, of the form
\begin{equation}
    W_n\left(x_1,\ldots,x_n;\{g_k\}\right) = \sum_{g=0}^{\infty}N^{2-2g-n} W_{g,n}\left(x_1,\ldots,x_n;\{g_k\}\right) \quad \text{as } \quad N \to \infty.
\end{equation}
The single trace function $W_{0,1}$ satisfies an algebraic equation, known as the spectral curve of the matrix model. 
Its Laplace transform, the so-called loop operator denoted $\widehat{W}_{0,1}$, may be expressed as
\begin{equation}\label{eq:linearloopeqn}
    \widehat{W}_{0,1}(t) = \lim_{N \to \infty} \frac{1}{N}\left\langle \tr{e^{tA}} \right \rangle,
\end{equation}
and solves an integral equation known as the loop equation\footnote{$\star$ denotes convolution. $(f \star g) (t) = \int_{0}^t \dif s \, f(s)g(t-s)$.}
\begin{equation}\label{eq:matrixmodelloop}
    V'\left(\dif/\dif t\right) \cdot \widehat{W}_{0,1}(t) = \left(\widehat{W}_{0,1} \star \widehat{W}_{0,1}\right)(t).
\end{equation}
In the present work we develop a generalisation of this matrix model. 
We study a $d$-dimensional Hermitian matrix model together with a smooth path $\gamma: [0,1] \to \mathbb{R}^d$ and replace the exponentials in the loop operator with path-ordered exponentials
\begin{equation}
    \widehat{W}_{0,1}(t) = \lim_{N \to \infty} \left\langle \tr{e^{t A}} \right \rangle \longrightarrow \widehat{W}_{0,\gamma}(t) = \lim_{N \to \infty} \left\langle \tr{\textrm{P}\exp \int_{\gamma} A } \right \rangle.
\end{equation}
In the following we derive corresponding loop equations generalising~\eqref{eq:matrixmodelloop} for these path loop operators.

\subsection{Wilson lines}\label{subsec:Wilsonlines}
We now turn to the study of smooth paths $\gamma: [0,T] \to \mathbb{R}^d$. 
In the context of learning and kernel methods, such paths could be, for example, trajectories of the price a collection of financial instruments or patient heart rate data. In the present context we think of them as the world line of a charged particle moving through a gauge field---a so-called \textit{Wilson line} in the physical terminology---we now explain this construction.

We consider $V = \mathbb{R}^d$ as a smooth Riemannian manifold with the Euclidean metric.\footnote{The following construction applies more generally to parallelisable manifolds.} We consider a (topologically trivial) principal $G=U_N$ bundle $P = V \times G \xrightarrow{\pi} V$.
A connection on $P$ is then uniquely specified by a Lie algebra valued one-form, written $A \in \Omega^1(V,\mathfrak{g})$, and we consider the case where $A$ is constant on $V$. A path $\gamma: [0,T] \to V$ then has a unique horizontal lift\footnote{A horizontal lift of a path $\gamma \in V$ is a path $\tilde{\gamma} \in P$ such that $\pi(\tilde{\gamma}) = \gamma$ and $\dot{\tilde{\gamma}}$ lies in the horizontal subspace $V \cong TV$.} $\tilde{\gamma}: [0,T] \to P$ where $\tilde{\gamma}(t) = (\gamma(t),g(t))$ is determined by the solution to the ODE
\begin{equation}\label{eq:paralleltransport}
    \dot{g}(t) = g(t) \cdot A(\dot{\gamma}(t)),\quad g(a) = e.
\end{equation}
We now consider a finite-dimensional $G$-representation $\mathcal{H}$. We will take throughout $G=U_N$ with $\mathcal{H} = \mathbb{C}^N$ the fundamental representation and therefore, in a minor abuse of notation, conflate $A$ with its representation. 
Following the Dirac notation convention, vectors in the associated bundle are written as $|\psi\rangle \in \mathcal{H}$. We then define the topologically trivial associated bundle by $\mathbb{H} = V \times \mathcal{H}$ and write the fibre at $x \in V$ as $\mathcal{H}_x$. The parallel transport of a vector $|\psi\rangle \in \mathcal{H}_{\gamma(a)}$ along $\gamma$ from $\gamma(0) \in V$ to $\gamma(T) \in V$ is then determined by the solution to~\eqref{eq:paralleltransport}. In coordinates this equation reads
\begin{equation*}
    \dif U_{0,t}(\gamma) = U_{0,t}(\gamma) \cdot A_{\mu} \dif\gamma^{\mu},\quad U_{0,0} = I_N,
\end{equation*}
and the solution to this equation defines the path-ordered exponential
\begin{equation}\label{eq:unitarydevelopment}
    U_{0,T}(\gamma) = \textrm{P}\exp\left( \sum_{\mu=1}^d\int_{0}^T A_\mu\dif\gamma^\mu\right).
\end{equation}
The setup is illustrated in \cref{fig:wilsonline}. Later, in \cref{sec:quantumsignature}, the associated bundle will be further interpreted physically as a $N=2^n$ dimensional Hilbert space of qubits $\mathcal{H} = (\mathbb{C}^2)^{\otimes n}$ and $|\psi\rangle$ will represent the state of a register of qubits. 
The development $U_{0,T}(\gamma)$ is then a time-dependent many-body time evolution operator. 

In the physical terminology $\gamma: [0,T] \to V$ describes the worldline of a charged particle moving from $\gamma(0)$ to $\gamma(T)$ through a (constant) electromagnetic field described by the connection $A$, 
that is a Wilson line. The charge of the particle is described by a vector in $\mathcal{H}$ and the parallel transport operator~\eqref{eq:unitarydevelopment} describes how the charge changes throughout the motion. The constant connection has non-trivial curvature and thus introduces a path-dependence in this charge evolution. In the data science context, this vector may thus be used as a feature associated to the path.

In the context of defining kernels on path space, we will later be interested in the expectation values of the trace of this quantity, $\tr{U_{0,T}(\gamma)}$. In summary, we are interested in integrating over the choice of connection using $d$-dimensional generalisations of the matrix model integrals discussed in the previous subsection.

\begin{figure}
    \centering
    \includegraphics[scale=0.3]{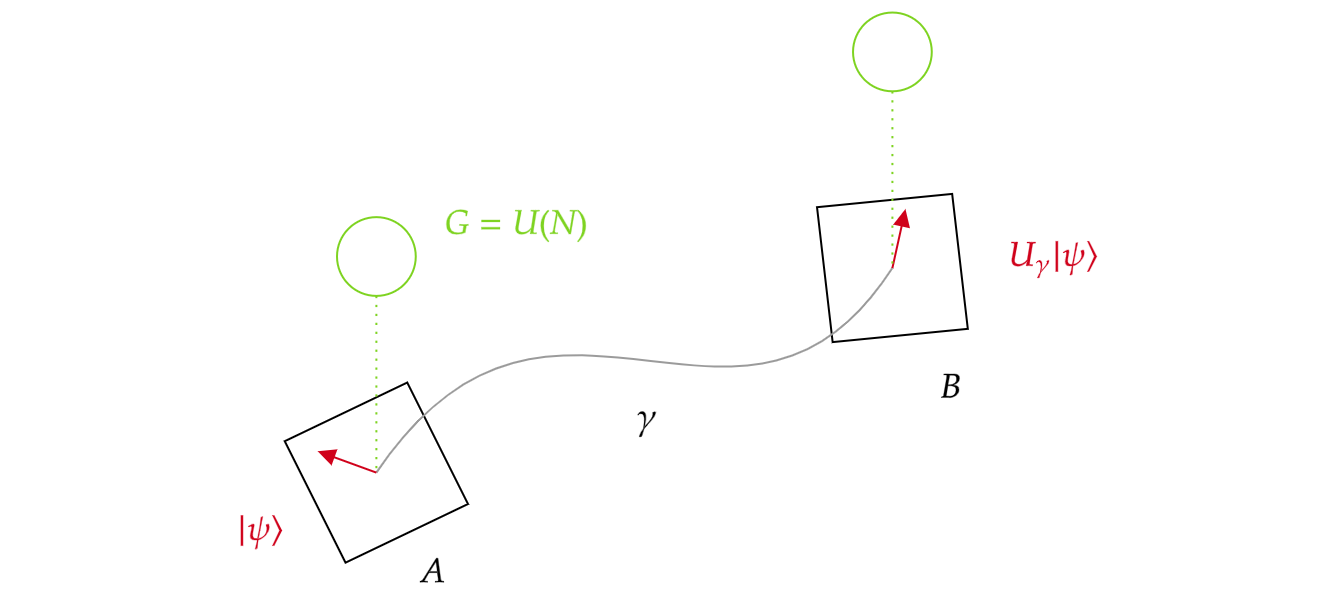}
    \caption{The path development as parallel transport in a $G=U(N)$ bundle along a Wilson line.}
    \label{fig:wilsonline}
\end{figure}

\begin{remark}
The quantum theory we define interpolates between zero-dimensional QFT and $d$-dimensional Yang-Mills theory. The unitary development $U_{0,T}(\gamma)$ may be considered more generally as a Wilson line in Yang-Mills theory on $\mathbb{R}^d$ where $A_{\mu}$ now varies over spacetime. Instead, we simplify the setup and consider a constant off-shell gauge field on a parallel manifold averaged over a matrix model.
\end{remark}

\subsection{Path signatures}
When analysing the path-ordered exponentials appearing in~\eqref{eq:unitarydevelopment}, it will be convenient to consider power series expansions. The terms involving the path $\gamma$ will be given by its iterated integrals, the collection of which is called its \emph{path signature}.\footnote{In fact, the path signature itself fits into the geometric setup of Section~\ref{subsec:Wilsonlines} as a parallel transport on a principal $T((V))$ bundle. The corresponding Lie algebra $\mathcal{L}(V)$ has a universal property meaning its $\mathfrak{g}$-representations are determined by homomorphisms $\phi: V \to \mathfrak{g}$. Hence connections on $G$-bundles determine representations of $T((V))$ uniquely.} Arising in the work of K.T. Chen in the 1950s \cite{Chen_2,chen1957integration} and forming a central part of Lyons' rough path theory \cite{lyons1998differential}, the signature has more recently found success both as a theoretical and as a practical tool in modern machine learning. The signature forms the analogue of the collection of monomials for vector valued data: under certain conditions, linear functionals of the signature are dense in continuous functions on compact subsets of (unparameterised) path space \cite{Levin_Lyons_Ni,CT_topologies_a}. Before we give the definition of a signature we require some basic definitions.

\begin{definition}[Words]\label{defn:words}
    For $d\in \N$, let $\WW_d$ denote the set of words in the $d$-letters $\{1,\ldots,d\}$ and let $\varnothing\in\WW_d$ denote the empty word. Define $\AA_d$ to be the algebra over $\WW_d$ with base field $\C$ and multiplication defined by 
    \[
    \bm{w}\bm{u}\coloneqq w_1\cdots w_n u_1\cdots u_m,
    \]
    for $\bm{w},\bm{u}\in\WW_d$. Define also the map $\abs{\cdot}:\WW_d\to\N$ to be the length of a word.
\end{definition}

\begin{definition}[Time-Simplex]
    For any $n\in\N$, define the $n$-simplex over an interval $[a,b]$ by
    \[
    \Delta_{[a,b]}^n\coloneqq \big\{(t_1,\ldots,t_n):a\leq t_1\leq\cdots\leq t_n\leq b\big\}.
    \]
\end{definition}

Throughout, we will take $V\equiv \R^d$ equipped with the standard Euclidean inner product. 
The signature of a path $\gamma:[0,T] \to V$ is then defined as follows.
\begin{definition}[Paths and Signatures]
    Let $\XX\coloneqq C^{\text{1-var}}\big([0,T], V\big)$ be the space of continuous $1$-variation paths $\gamma:[0,T]\to\R^d$ with $\gamma_0=0$. For $\gamma\in\XX$, the signature of $\gamma$ is a function
    \[
    \SS_{\cdot,\cdot}(\gamma):\Delta_{[0,T]}^2\to T((V))\coloneqq \prod_{m=0}^\infty V^{\otimes m},
    \]
    where $\SS_{s,s}(\gamma) = (1,0,\ldots)$ for all $s\in [0,T]$ and
    \[
    \SS_{s,t}(\gamma)^m\coloneqq \int_{\Delta_{[s,t]}^m}\dif\gamma_{t_1}\otimes\cdots\otimes \dif \gamma_{t_m}\in V^{\otimes m},
    \]
    for $m\geq 1$ with $\SS_{s,t}(\gamma)^0\coloneqq 1\in\R$. For a word $\bm{w}\in\WW_d$ with $|\bm{w}|=m$, we also define the coordinate iterated integral
    \[
    \SS_{s,t}^{\bm{w}}(\gamma)\coloneqq \int_{{\Delta_{[s,t]}^m}}\dif\gamma_{t_1}^{w_1}\cdots \dif \gamma_{t_m}^{w_m}\in \R,
    \]
    where $\SS_{s,t}^{\varnothing}(\gamma)\coloneqq \SS_{s,t}(\gamma)^0=1.$ 
    We will commonly write $\SS(\gamma)\coloneqq \SS_{0,T}(\gamma)$ for the signature over the entire interval.
\end{definition}
\begin{proposition}[Chen's identities]
    We define the concatenation of two paths $\gamma,\sigma\in\XX$ by 
    \[
    (\gamma\star\sigma)_t\coloneqq \begin{cases}
        \gamma_{2t},\quad &t\in [0,T/2],\\
        \gamma_{T}+\sigma_{2(t-T/2)},\quad &t\in [T/2,T].
    \end{cases}
    \]
    Then
    \[
    \SS(\gamma\star\sigma)=\SS(\gamma)\otimes\SS(\sigma).
    \]
    Additionally, we may define the reversal of the path $\gamma$ by
    \[
    \overleftarrow{\gamma}_t\coloneqq \gamma_{T-t}-\gamma_{T}.
    \]
    Then
    \[
    \SS(\gamma\star\overleftarrow{\gamma})=\SS(\gamma)\otimes\SS(\overleftarrow{\gamma})=\mathbf{1}.
    \]
\end{proposition}
Chen's identities imply that the space of signatures $\SS(\XX)\subset T((V))$ forms a group with multiplication given by path concatenation and inverse given by path reversal. Using the signature, we may expand~\eqref{eq:unitarydevelopment} as
\begin{equation}
    U_{0,T}(\gamma) = \textrm{P}\exp\left( \sum_{i=1}^d\int_{0}^T A_i\dif\gamma_u^i\right)=\sum_{\bm{w}\in\WW_d}\mathrm{i}^{|\bm{w}|}A_{\bm{w}}\SS_{a,b}^{\bm{w}}(\gamma),
\end{equation}
where $\mathrm{i}$ denotes the imaginary unit. In this work, particularly later in the context of kernels, we are interested in computing the expectation value of this quantity over the matrix models in the large $N$ limit discussed in the previous subsection. Provided the quantities exist, we introduce the notations

\begin{align}
    \langle \gamma_{0,T}\rangle_{\mu^N} &\coloneqq \sum_{\bm{w}\in\WW_d}i^{|\bm{w}|} \mathbb{E}_{\mu^N}\big[\tr{A_{\bm{w}}}\big]\SS_{a,b}^{\bm{w}}(\gamma),\label{eq:tracedevelopment}\\
    \langle \gamma_{0,T}\rangle_{\mu^\infty} &\coloneqq \sum_{\bm{w}\in\WW_d}i^{|\bm{w}|} \lim_{N\to\infty}\mathbb{E}_{\mu^N}\big[\tr{A_{\bm{w}}}\big]\SS_{a,b}^{\bm{w}}(\gamma).\label{eq:tracepathdevelopment}
\end{align}
Finally, we note that in the case $d=1$ when $\gamma_t = t$ is a straight line, the signature becomes $\SS_{0,T}(\gamma)^n = T^n/n!$. The right-hand side of~\eqref{eq:tracedevelopment} then becomes the matrix model moment generating function of trace correlators discussed in the previous section
\begin{equation*}
    \langle 1_{0,T} \rangle_{\mu^N} = \sum_{n=0}^{\infty}\langle\tr{A^n}\rangle_{\mu^N} \frac{T^n}{n!} = W_N(T).
\end{equation*}
The present article may then be considered a generalisation of traditional Physics matrix models to the case of~$d$ multi matrices and the observable is upgraded to a Wilson line.

\subsection{Signature kernels}\label{subsec:sigkernels}
We now discuss kernels on path space which are important from an applications perspective. Typically, one considers a feature map $\varphi: \mathcal{X} \to \mathcal{H}$ that embeds the data $\mathcal{X}$ into a typically high-dimensional (possibly infinite) Hilbert space $\mathcal{H}$ where linear methods may be applied. A kernel $k:\mathcal{X} \times \mathcal{X} \to \mathbb{R}$ is simply the inner product of the feature map, namely
\begin{equation}\label{eq: reg_kernel}
    k(x,y)\coloneqq \langle \varphi(x),\varphi(y)\rangle_{\HH}.
\end{equation}
If one can access the Gram matrix $G_{ij}=k(x_i,x_j)$ for all pairs of points in the dataset, then one can perform a variety of linear separation techniques without ever evaluating $\varphi$. Such an ability is known as a ``kernel trick''. 
Conversely, every kernel corresponds to some feature map into some Hilbert space, known as the reproducing kernel Hilbert space (RKHS). Indeed let $\HH^0=\text{span}\big\{k(x,\cdot):x\in\XX\big\}$, and for functions
\[
f=\sum_{i=1}^n\alpha_ik(x_i,\cdot)\quad\text{and}\quad g=\sum_{j=1}^mk(y_j,\cdot),
\]
define
\[
 \langle f,g\rangle_{\HH^0}\coloneqq\sum_{i=1}^n\sum_{j=1}^m\alpha_i\beta_jk(x_i,y_j),
\]
then $\HH\coloneq \overline{\HH^0}$ with respect to this inner product. This space satisfies the reproducing property
\[
\langle f, k(x,\cdot)\rangle_{\HH}=f(x),\quad\text{for all }f\in\HH\text{ and }x\in\XX,
\]
and the feature map $\varphi$ is simply given by the embedding $x\mapsto k(x,\cdot)$. By equipping (an appropriate subspace of) $T((V))$ with the standard Hilbert-Schmidt inner product, we arrive at the ordinary signature kernel defined by
\[
k_{s,t}^{\text{sig}}(\gamma,\sigma)\coloneqq \langle \SS_{0,s}(\gamma),\SS_{0,t}(\sigma)\rangle\coloneqq \sum_{n=0}^\infty \langle \SS_{0,s}(\gamma)^n,\SS_{0,t}(\sigma)^n\rangle_{V^{\otimes n}},
\]
which is of the form~\eqref{eq: reg_kernel} with feature map $\varphi(\cdot)=\SS(\cdot)$. The ordinary signature kernel comes with a kernel trick, being that $k_{s,t}^{\text{sig}}$ solves the integral equation \cite{SigPDE}
\begin{equation}\label{eq: sig_kernel}
k_{s,t}^{\text{sig}}=1+\int_0^s\int_0^tk_{u,v}^{\text{sig}}\langle \dif \gamma_u,\dif\gamma_v\rangle.
\end{equation}
A variety of algorithms \cite{lemercier_scheme,Piatti_scheme,SigPDE} have been proposed to solve~\eqref{eq: sig_kernel}, however, they all suffer from quadratic complexity in the length of the input time series. An alternative approach is to compute the inner product between low-dimensional random projections of the signature that approximately preserve their inner products. Indeed, consider the differential equation
\[
\dif Z^\gamma_t = \sum_{i=1}^dZ_t^{\gamma}A_i\dif\gamma^i_t,\quad Z_0^\gamma= I_N,
\]
where each $A_i\in\R^{N\times N}$ is a random matrix with i.i.d. $\NN\big(0, \tfrac{1}{N}\big)$ entries and $N$ is large. Then $Z^\gamma$ is  a $\text{GL}_N(\R)$-valued random variable where
\[
\langle Z_{s}^\gamma, Z_{t}^\sigma\rangle\coloneqq \frac{1}{N}\Tr{(Z_{s}^{\gamma})^\top Z_{t}^\sigma} \overset{N\to\infty}{\rightarrow}k^{\text{sig}}_{s,t}(\gamma,\sigma),
\]
see for example \cite{SKlimit,CT_free}. Solving for $Z^\gamma$ is then only linear in the length of the input time series. Supposing that $\gamma=\gamma_{1}\star\cdots\star\gamma_{n}$ is a piecewise linear paths valued in $\mathbb{R}^d$ with linear increments $\gamma_i$ for $i=1,\ldots,n$, then $Z_T^\gamma$ is given exactly by the matrix product
\[
Z_T^\gamma=\exp\Big(\sum_{i=1}^dA_i\gamma_1^i\Big)\cdots \exp\Big(\sum_{i=1}^dA_i\gamma_n^i\Big).
\]
Clearly, this exact solution becomes costly to compute when~$N$ is large. An approach to circumvent this issue might be to consider sparse approximations of the~$A_i$. However, even for sparse~$A_i$, each matrix exponential is typically dense and so the products are still an obstacle. Supposing however, that the matrices $A_i$ where instead anti-Hermitian, then each matrix exponential would be a unitary matrix with sparse representation. Such a setting naturally leads one to consider whether quantum algorithms are applicable, which is the study of \cref{subsec:quantumalg}.

In alignment with the previous discussion, in the present work we are interested in kernels on paths induced by unitary matrix path developments, that is the Wilson line observables introduced above averaged over matrix models with potential $V$. The Lie group $U_N$ has a natural unitary inner product (the Hilbert-Schmidt inner product) and we may define
\begin{equation}
    k^{\mu_V^N}(\gamma,\sigma) := \langle\tr{U_{\gamma}U_{\sigma}^{\dag}}\rangle_{\mu_V^N}\coloneqq \frac{1}{N}\langle\Tr{U_\gamma U_\sigma^{\dag}}\rangle_{\mu_V^N},
\end{equation}
where we have replaced $Z^\gamma$ with $U_\gamma$ to reflect the unitary nature of the developments. 
The particular case where~$\mu_V$ is the free Gaussian multi-matrix model, i.e. the Gaussian Unitary Ensemble (GUE), will be the focus of \cref{sec:quantumsignature} where we will discuss its interpretation as the kernel associated to a quantum feature map built from an efficient quantum circuit. Further, in the case that $\mu_V$ is a free Gaussian multi-matrix model, a governing integral for
\[
k^{\text{GUE}}(\gamma,\sigma)\coloneqq \lim_{N\to\infty}k^{\mu_V^N}(\gamma,\sigma)=\lim_{N\to\infty}\langle\tr{U_{\gamma}U_{\sigma}^{\dag}}\rangle_{\mu_V^N}
\]
was derived in \cite{CT_free}. The derivation of the limiting functional equation relied heavily on the Schwinger-Dyson equation that is asymptotically satisfied by the GUE ensemble. \cref{sec:unitarydevelopments} is devoted to finding the corresponding functional equations in the case of a small perturbation of the Gaussian law. Such perturbations asymptotically satisfy Schwinger-Dyson equations that are perturbations of the one asymptotically satisfied by the GUE ensemble. Kernels like $k^{\text{GUE}}$ are examples of kernels where the feature map is not known explicitly, and whose existence is only guaranteed through the construction of its RKHS.

\paragraph{Geometric interpretation.}
The kernel also enjoys a geometric interpretation in the context of \cref{subsec:Wilsonlines}, the setup is illustrated in \cref{fig:PT}. 
We begin with a charge vector $|\psi\rangle$ living in the associated $N$~dimensional $U_N$ bundle $\mathcal{H}$. In Euclidean space $\mathbb{R}^d$, we may identify $\mathcal{H}_{\tau(0)}$ with $\mathcal{H}_{\sigma(0)}$ and $\mathcal{H}_{\tau(T)}$ with $\mathcal{H}_{\sigma(T)}$. The kernel is then a measure of the difference between a charge vector $|\psi\rangle$ being transported along $\tau$, compared with the same vector being transported along $\sigma$. That is we may compare the angles of $U_{\tau}|\psi\rangle$ with $U_{\sigma}|\psi \rangle$ by taking the Hermitian inner product. If we average the resulting product uniformly over a basis $\{|\psi_i\rangle\}$ of $\mathcal{H}$ then we find
\begin{equation*}
    \frac{1}{N}\sum_{i=1}^N \langle \psi_i | U_{\sigma}^{\dag}U_{\tau} | \psi_i \rangle = \frac{1}{N}\tr{U_{\sigma}^{\dag} U_{\tau}},
\end{equation*}
the right-hand side is the definition of the kernel as the Hilbert-Schmidt inner product and thus we have a geometrical interpretation of the path development signature kernel as the average change in angle between charge vectors under two Wilson lines. 

\begin{figure}
    \centering
    \includegraphics[width=0.75\linewidth]{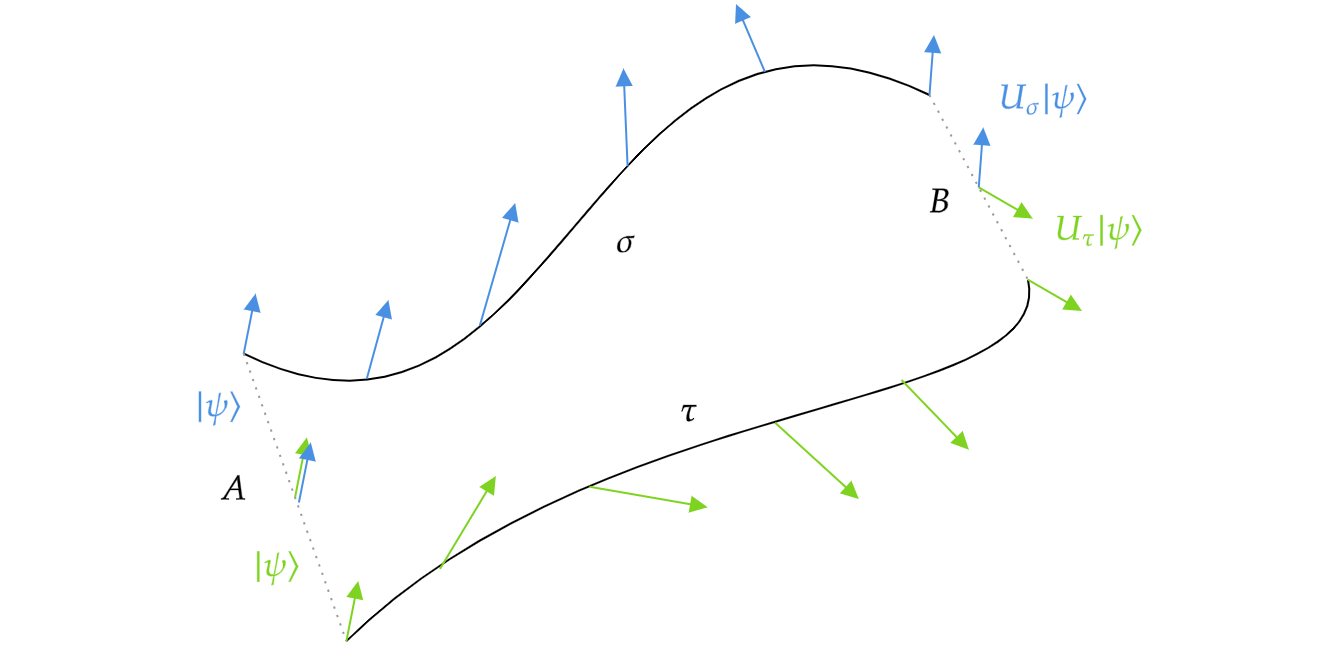}
    \caption{The kernel as parallel transport along $\tau$ compared with parallel transport along $\sigma$.}
    \label{fig:PT}
\end{figure}

\section{Universal Limit of Matrix Model Unitary Developments}\label{sec:unitarydevelopments}
This section is dedicated to deriving integro-differential equations satisfied by $\pdg{s}{t}$ for small perturbations of the Gaussian potential. We start by relating this to the quantity $k^{\mu_V}(\tau,\sigma)$. By standard properties of path developments, $U_\tau U_\sigma=U_{\tau\star\sigma}$ and $U_\sigma^{-1}=U_{\overleftarrow{\sigma}}$. Since the inverse of a unitary matrix is its conjugate transpose, it holds that
\[
k^{\mu_V}(\tau,\sigma)=\langle\tr{U_{\tau\star\overleftarrow{\sigma}}}\rangle_{\mu_V}
\]
As such, it is enough to understand the limiting functional $\pdg{s}{t}$ for a single path $\gamma$ and replace it by $\gamma=\tau\star\overleftarrow{\sigma}$ when conducting a kernel evaluation, which is what we do in this section. The case where $V$ corresponds to the Gaussian potential $\tfrac{1}{2}\sum_{i=1}^dX_i^2$ was considered in~\cite{CT_free} along with variants leading to the same limit. In the present work we extend these results to the matrix models discussed above in \cref{sec:background}. Recall that $\mathfrak{u}_N=i\mathfrak{h}_N$, where $\mathfrak{h}_N$ is the set of $N\times N$ Hermitian matrices and that for a path $\gamma\in\XX$, the unitary path development of path segment $\gamma_{s,t}$ with development map $M\in\text{Hom}(V, \mathfrak{u}_n)$ is the time $t$-solution to the $\MM_N(\C)$-valued controlled differential equation
\begin{equation}\label{eq: pcf_eq}
    \dif U_{s,u} = U_{s,u}\cdot M(\dif \gamma_u),\quad Z_{s,s} = I_N,
\end{equation}
where $\cdot$ denotes matrix multiplication. Fixing a basis for $V \cong \mathbb{R}^d$ and noting that $M(v)=i\sum_{j=1}^dA_jv^j$ for matrices $A_j\in\mathfrak{h}_N$, the differential equation reads
\begin{equation}
    \dif U_{s,u}=\sum_{j=1}^d U_{s,u}A_j\dif\gamma_u^j,
\end{equation}
whose solution is given by
\[
U_{s,t} = \textrm{P}\exp\left( i\sum_{j=1}^d\int_{s}^t A_j\dif\gamma_u^j\right)=\sum_{\bm{w}\in\WW_d}i^{|\bm{w}|}A_{\bm{w}}\SS_{s,t}^{\bm{w}}(\gamma),
\]
where $\textrm{P}\exp$ is the path-ordered exponential along the curve $\gamma$. We may randomise the solution to~\eqref{eq: pcf_eq} via the randomisation of the collection $(A_1,\ldots,A_d)\sim\mu_V^N$. In particular, our quantity of interest is $\lim_{N\to\infty}\Ex\big[\tr{U_{s,t}}\big]$.

\subsection{Non-commutative laws and Schwinger-Dyson equations}\label{sec: sd_intro}
In this section, we gather various notions both from  the world of non-commutative probability theory and from rough path theory, in particular following~\cite{Guionnet_SD} for the non-commutative definitions.
We shall use~$\overline{\ \cdot\ }$ to denote complex conjugation
\begin{definition}[$C^*$ algebra]
A $C^*$ algebra $\AA$ is a Banach algebra together with an involution $\ast:\AA\to\AA$ satisfying the following properties
\begin{enumerate}[label=\arabic*)]
    \item $(a^{\ast})^\ast=a$;
    \item $(a+b)^{\ast}=a^{\ast}+b^\ast$;
    \item $(ab)^\ast=b^\ast a^\ast $;
    \item $(\lambda a)^\ast=\bar{\lambda}a^{\ast}$ for all $\lambda\in\C$;
    \item $\norm{a^{\ast}}=\norm{a}$ and $\norm{a^{\ast}a}=\norm{a}^2$.
\end{enumerate}
\end{definition}
\begin{definition}
    Let $X_1,\ldots,X_d$ be non-commutative indeterminates. Define $\CC_d\coloneqq\C\langle X_1,\ldots,X_d\rangle$ to be the space of polynomials in the indeterminates with complex coefficients. Let $\ast$ be an involution on $\CC_d$ such that the indeterminates are self-adjoint with respect to $\ast$, i.e. $X_i^\ast=X_i$. Given a word $\bm{w}\in\WW_d$ we will write $X_{\bm{w}}\coloneqq X_{w_1}\cdots X_{w_n}$, which may be extended by linearity to all $\bm{w}\in\AA_d$ with $X_{\varnothing}=1$. This induces an algebra isomorphism $\AA_d\to\CC_d$.
\end{definition}
\begin{definition}[Non-commutative law]\label{def: non-comm-law}
    A non-commutative law $\tau\in\CC_d^*$ is a linear functional on $\Cd$ taking values in $\C$ such that
    \begin{enumerate}[label=\arabic*)]
        \item Positivity: for all $P\in\Cd$, $\tau(P^\ast)=\overline{\tau(P)}$ and $\tau(PP^\ast)\geq 0$.
        \item Mass: $\tau(1)=1$.
        \item Trace: for all $P,Q\in\Cd$, $\tau(PQ)=\tau(QP)$.
    \end{enumerate}
\end{definition}
A general non-commutative probability space is a pair $(\AA,\varphi)$ of a unital $C^*$ algebra $\AA$ equipped with a sate $\varphi$. However, given $d$ non-commuting random variables $a_1,\dots,a_d$ living in $(\AA,\varphi)$, there exists a unique non-commutative law $\tau$ on the free $*$-algebra $\CC_d$ such that $\tau\big(P(X_1,\dots,X_d)\big)=\varphi\big(P(a_1,\dots,a_d)\big)$ for any polynomial $P$. As such, \cref{def: non-comm-law} is enough for our purposes.
\begin{definition}
    A sequence $(\tau_n)_{n=1}^\infty$ of  non-commutative laws is said to converge weakly to $\tau\in\CC_d$ if
    \[
    \lim_{n\to\infty}\tau_n(P)=\tau(P),\text{ for all }P\in\CC_d.
    \]
\end{definition}
It can readily be checked that the properties defining a non-commutative law are closed under weak convergence, and so any weak limit of non-commutative laws is itself a non-commutative law. It turns out that several important non-commutative laws are characterised by various integration by parts formulae. To motivative this perspective, we mention the famous example from the commutative world: the Gaussian integration by parts formula. Consider a vector $(X_1,\ldots,X_d)\sim \NN(0,I_d)$ of commutative Gaussian random variables, and a sufficiently smooth function $f:\R^d\to\R$. 
Then it is known that
\begin{equation}\label{eq: gaussian_ibp}
\Ex[X_if(X_1,\ldots,X_d)]=\Ex[\partial_if(X_1,\ldots,X_d)].
\end{equation}
Now, taking $f$ to be some (commutative) monomial $X_i^nX_{\bm{w}}$, where $\bm{w}$ is devoid of the letter $i$, then 
\begin{equation}\label{eq: ibp_to_wicks}
\Ex[X_i^{n+1}X_{\bm{w}}]=n\Ex[X_i^{n-1}X_{\bm{w}}].
\end{equation}
This inductive relation is exactly the defining property of Wick's formula:
\[
\Ex[X_{\bm{w}}]=\sum_{\pi\in \PP_2^{|\bm{w}|}} \prod_{{(i,j)\in\pi}}\Ex[X_{w_i}X_{w_j}],
\]
where $\PP_2^{n}$ denotes the set of pair partitions on $n$ letters. Indeed, the factor of $n$ in~\eqref{eq: ibp_to_wicks} is exactly the number of ways of pairing the new $X_i$ with one $X_i$ from $X_i^n$. Thus, combinatorial characterisation of the Gaussian law, Wick's formula, and the integration by parts characterisation,~\eqref{eq: gaussian_ibp}, are equivalent. To extend this observation to non-commutative laws we first need suitable notions of non-commutative derivatives. In the following, $\CC_d\otimes\CC_d$ is the usual tensor product of vector spaces.

\begin{definition}[Free Difference Quotient]
    For $d\in\N$ and $1\leq i \leq d$ we define the free difference quotient $\del_i:\CC_d\to\CC_d\otimes \CC_d$ by
    \[
    \del_iPQ \coloneqq \del_i P \times \big(1 \otimes Q\big)+\big(P\otimes1\big)\times \del_iQ,
    \]
    with $P\otimes Q\times R\otimes S\coloneqq PR\otimes QS$ and $\del_{i}X_j\coloneqq(1\otimes 1)\delta_{ij}$. In particular, for a word $\bm{w}\in \WW_d$ we have
    \begin{equation}\label{eq: free_diff_monomial}
    \del_i X_{\bm{w}}=\sum_{\bm{w}=\bm{u}i\bm{v}}X_{\bm{u}}\otimes X_{\bm{v}}.
    \end{equation}
\end{definition}
\begin{definition}[Cyclic Derivative]\label{def: cyclic}
    Define $\mathfrak{m}:\CC_d\otimes\CC_d\to\CC_d$ by $\mathfrak{m}(P\otimes Q)\coloneqq QP$, then the cyclic derivative $D_i : \CC_d \to \CC_d$ is given by
    \[
    D_iP\coloneqq (\mathfrak{m}\circ \del_i)P.
    \]
    In particular, for a word $\bm{w}\in \WW_d$, we have
    \[
    D_iX_{\bm{w}}=\sum_{\bm{w}=\bm{u}i\bm{v}}X_{\bm{v}}X_{\bm{u}}.
    \]
\end{definition}
Via the bijection $\CC_d\to\AA_d$, the operators $\del_i$ and $D_i$ induce well-defined maps on $\AA_d\to\AA_d\otimes\AA_d$ and $\AA_d\to\AA_d$ respectively. As such, we will often write $\del_iX_{\bm{w}}=X_{\del_i\bm{w}}$ or $D_iX_{\bm{w}}=X_{D_i\bm{w}}$. We are now ready to give the non-commutative version of an integration by parts rule.
\begin{definition}[Schwinger-Dyson Equations]
    Let $\tau:\CC_d\to\C$ be a non-commutative law. We say that $\tau$ satisfies a Schwinger-Dyson equation with conjugates $P_i\in \CC_d$ if, for all $1\leq i\leq d$,
    \begin{equation}\label{eq: SD}
        \tau\otimes\tau\big(\del_i P\big)=\tau\big(PP_i\big),
    \end{equation}
    where $\tau\otimes\tau(P\otimes Q)\coloneqq \tau(P)\tau(Q)$.
\end{definition}
Perhaps the most well-known and important example of a Schwinger-Dyson equation is the one satisfied by freely independent semicircular random variables. Free independence is a non-commutative notion of independence, where joint moments are determined by non-crossing partitions rather than the usual product rule for classically independent random variables. To illustrate this, suppose that $X_1$ and $X_2$ are freely independent with mean zero and law $\tau$, then $\tau(X_1X_2X_1X_2)=0$, whereas in the commutative world with ordinary independence, it would be equal to $\Ex[X_1^2]\Ex[X_2^2]\neq 0$ as long as $X_1$ and $X_2$ are non-trivial.
\begin{example}[Semicircular Variables]\label{ex: semicircular}
    The law~$\tau$ of $d$-free semicircular random variables satisfies the Schwinger-Dyson equation
    \[
     \tau\otimes\tau\big(\del_i P\big)=\tau\big(PX_i\big).
    \]
    And so for a word $\bm{w}\in\WW_d$, 
    \[
    \tau(X_{\bm{w}i})=\sum_{\bm{w}=\bm{u}i\bm{v}}X_{\bm{u}}X_{\bm{v}}.
    \]
    Iterating this expression gives the analogue of Wick's formula for freely independent semicircular random variables:
    \[
    \tau(X_{\bm{w}})=\sum_{\pi\in \mathrm{NC}_2^{|\bm{w}|}}\prod_{(i,j)\in\pi}\tau(X_{w_i}X_{w_j}),
    \]
    where $\mathrm{NC}_2^n$ is the set of \emph{non-crossing} pair partitions of $n$ letters.
\end{example}
An important class of Schwinger-Dyson equations, including \cref{ex: semicircular}, are those where the conjugate polynomials $P_i$ are given by the cyclic derivatives of some polynomial $V\in\CC_d$. In this case, we have
\begin{equation}\label{eq: SD-V}
    \tau\otimes\tau(\del_iP)=\tau(PD_iV).
\end{equation}
The non-commutative laws satisfying the Schwinger-Dyson equation of type~\eqref{eq: SD-V} arise as the asymptotic limit of the non-commutative laws of interacting random matrix ensembles. For a self-adjoint polynomial $V\in\CC_d$, consider the law $\mu_V^N$ on $N\times N$ Hermitian matrices given by
\[
\dif\mu_V^N(X_1^N,\ldots,X_d^N)=\frac{1}{Z_V^N}\exp\Big\{-N\Tr{V(X_1^N,\ldots,X_d^N)}\Big\}\dif X^N,
\]
where
\[
\dif X^N=\prod_{1\leq j\leq k\leq N}\dif\mathfrak{R}\big(X^N(jk)\big)\prod_{1\leq j<k\leq N}\dif\mathfrak{I}\big(X^N(jk)\big),
\]
and $Z_V^N$ is a normalising constant. In order to ensure the partition function $Z_V^N$ is finite for all $N$, some restrictions on $V$ are required. Assume that $V$ takes the form $ V=\frac{1}{2}\sum_{i=1}^dX_i^2+ W$ for a self-adjoint polynomial $W=\sum_{i=1}^mg_iX_{\bm{w}_i}$. We say that $W$ is $c$-convex if the map
\begin{equation}\label{eq: convex_V}
    \big(X_\nu^N(jk)\big)_{\substack{1\leq j\leq k\leq N \\ 1\leq \nu\leq d}}\mapsto \Tr{W(X^N)} + \frac{1-c}{2}\sum_{i=1}^d\Tr{X_i^2}
\end{equation}
is convex for all $N$ when viewed as a function of $\big(\R^{N^2}\big)^d$. Denote by $U_W$ the set of $\mathbf{g}=(g_1,\ldots,g_d)\subseteq \C^d$ for which~\eqref{eq: convex_V} holds. For $\varepsilon>0$ we also write $B_{\varepsilon}$ for the set of coefficients $\mathbf{g}$ for which $\norm{\mathbf{g}}_{\infty}\leq\varepsilon$. Recalling that $\tr{\cdot}\coloneqq\frac{1}{N}\Tr{\cdot}$ defines the normalised trace, we summarise now a collection of results from \cite{comb_matrix_models,matrix_models}, see also \cite[Chapter 7]{Guionnet_SD}.
\begin{theorem}[\cite{comb_matrix_models,matrix_models}]\label{thm: S-D_uniqueness}
    For every self-adjoint $c$-convex $W=\sum_{i=1}^mg_iX_{\bm{w}_i}$ and $R>2$, there exists $\varepsilon>0$ such that for all $\mathbf{g}\in U_W\cap B_\varepsilon$
    \[
    \lim_{N\to\infty}\Ex_{\mu_V^N}\big[\tr{X_{\bm{w}}^N}\big]=\tau_V(X_{\bm{w}}),
    \]
    for every $\bm{w}\in\WW_d$ where $\tau_V$ solves the Schwinger-Dyson equation~\eqref{eq: SD-V}. 
    Moreover, there is exactly one solution $\tau_V\in\CC_d^*$ to~\eqref{eq: SD-V} and it satisfies
    \[
    |\tau_V(X_{\bm{w}})|\leq R^{|\bm{w}|}\text{ for all }\bm{w}\in\WW_d.
    \]
\end{theorem}
Note that when $W=0$, the distribution $\mu_N^V$ is the Gaussian Unitary Ensemble (GUE) and the limiting law $\tau_V$ is the law of $d$ free semicircular random variables. From now on, we will only consider those $V$ satisfying the assumptions of \cref{thm: S-D_uniqueness}. In the following, let $\lambda_{\max}^N(A)$ be the maximum of the absolute values of the eigenvalues of an $N\times N$ matrix $A$, and let $\lambda_{\max}^N(\mathbf{A})$ be the maximum of $\big(\lambda_{\max}^N(A_i)\big)_{i=1}^d$. We will need the following lemma to be able to exchange certain limits and expectations.
\begin{lemma}[Sub-Gaussian tail of largest eigenvalue]\label{lem: sub_gaussian_tail}
    Let $V$ be a potential satisfying the conditions of $\cref{thm: S-D_uniqueness}$, and for each $N$ suppose that $(A_i)_{i=1}^d\sim \mu_V^N$. 
    Then there exist $\alpha>0$ and $t_0<\infty$ such that, for every $t\geq t_0$ and $N$,
    \[
    \P_{\mu_V^N}\big[\lambda_{\max}^N(\mathbf{A})\geq t\big]\leq e^{-\alpha N t^2}.
    \]
\end{lemma}
\begin{proof}
    See \cref{app: proofs}.
\end{proof}
\subsection{Path loop equations}
We now use the Schwinger-Dyson equations to derive governing equations for the trace of the path development in the large~$N$ limit.
Just as we required a version of the derivative for functions in non-commutative indeterminates, 
we will also need the notion of derivatives of functionals on path space. There are many such choices, but we will use the Dupire or vertical derivative~\cite{DupireDerivative}. 
This notion of derivative has been popular in deriving Taylor-like expansions of functions on path space.
\begin{definition}[Vertical derivative]
    Let $f:\XX\to\R$ be a suitably regular functional, $1\leq i\leq d$, and $h^i\in \XX$ the path defined by $t\mapsto \tfrac{ht}{T} e^i$ with $\{e^i\}_{i=1}^d$ basis vectors in $\mathbb{R}^d$. The vertical derivative of $f$ at a path $\gamma$ along the direction $i$ is defined by
    \[
    \nabla^i f(\gamma)\coloneqq \lim_{h\to 0}\frac{f\big(\gamma\star h^i\big)-f(\gamma)}{h}.
    \]
    For a word $\bm{w} = w_1\cdots w_k \in\WW_d$ of length $k$, define
    \[
    \nabla^{\bm{w}}\coloneqq \nabla^{w_1}\cdots\nabla^{w_k}f(\gamma).
    \]
\end{definition}
\begin{proposition}[Proposition~3.8 \cite{DupireDerivative}]
    For any path $\gamma\in\XX$ and $\bm{v},\bm{w}\in\WW_d$, the vertical derivative satisfies
    \[
    \nabla^{\bm{v}}\SS_{s,t}^{\bm{w}}(\gamma)=
    \begin{cases}
        \SS_{s,t}^{\bm{u}}(\gamma),\quad &\bm{w}=\bm{u}\bm{v},\\
        0,\quad &\text{otherwise.}
    \end{cases}
    \]
\end{proposition}
Consider now  $V\in \CC_d$ given by
\[
V = \frac{1}{2}\sum_{i=1}^dX_i^2 + W =\frac{1}{2}\sum_{i=1}^dX_i^2 + \sum_{j=1}^mg_jX_{\bm{v}_j},
\]
satisfying the conditions of \cref{thm: S-D_uniqueness} with associated non-commutative law $\tau_V$. Also associate to $W$ the path-dependent differential operators
\begin{equation}
    \DD_W^k \coloneqq -\sum_{j=1}^m(-i)^{|\bm{v}_j|}g_j\nabla^{D_k\bm{v}_j},
\end{equation}
where $D_k$ is the cyclic derivative from \cref{def: cyclic}. Recall now the definitions
\begin{align*}
    \langle \gamma_{a,b}\rangle_{\mu_V^N} &\coloneqq \sum_{\bm{w}\in\WW_d}i^{|\bm{w}|} \mathbb{E}_{\mu_V^N}\big[\tr{A_{\bm{w}}}\big]\SS_{a,b}^{\bm{w}}(\gamma),\label{eq:tracedevelopment}\\
    \langle \gamma_{a,b}\rangle_{\mu_V^\infty} &\coloneqq \sum_{\bm{w}\in\WW_d}i^{|\bm{w}|} \lim_{N\to\infty}\mathbb{E}_{\mu_V^N}\big[\tr{A_{\bm{w}}}\big]\SS_{a,b}^{\bm{w}}(\gamma).
\end{align*}
\begin{proposition}\label{prop: limits_interchange}
    Suppose that $V$ satisfies the assumptions of \cref{thm: S-D_uniqueness} and for each $s\in [0,T]$ let $U_{s,u}$ be the $U_N$-valued solution to the controlled differential equation
    \[
    \dif U_{s,u}=i\sum_{j=1}^d U_{s,u} A_j\dif\gamma_u^j,\quad U_{s,s}=I_N,
    \]
    where $(A_1^N,\ldots A_d^N)\sim \mu_V^N$. 
    Then
    \begin{equation}\label{eq: limiting_dev}
\Ex_{\mu_V^N}\big[\tr{U_{s,t}}\big]=\pdgn{s}{t}<\infty\quad\text{and}\quad\lim_{N\to\infty}\Ex_{\mu_V^N}\big[\tr{U_{s,t}}\big]=\pdg{s}{t}<\infty.
\end{equation}
\end{proposition}
The content of the above is simply a statement that limits and expectations may be interchanged with the sum over words.

\begin{proof}
    The expansion of the path ordered exponential in terms of the signature provides that
    \[
    \tr{U_{s,t}^N}=\sum_{\bm{w}\in\WW_d}i^{|\bm{w}|}\tr{{A_{\bm{w}}^N}}\SS_{s,t}^{\bm{w}}(\gamma).
    \]
    To attain both parts of~\eqref{eq: limiting_dev}, we simply need to justify the interchange of the limits and expectation with the sum over words $\bm{w}$. To do this, it is enough to show that
    \[
    \sum_{\bm{w}\in\WW_d}\Ex_{\mu_V^N}\Big[\big|\tr{A_{\bm{w}}^N}\big|\Big]\big|\SS_{s,t}^{\bm{w}}(\gamma)\big|
    \]
    is uniformly bounded in $N$. The sub-Gaussian nature of the tail behaviour of $\lambda_{\max}^N(\mathbf{A})$, \cref{lem: sub_gaussian_tail}, yields the following bound
    \begin{equation}
\Ex_{\mu_V^N}\Big[\big|\tr{A_{\bm{w}}^N}\big|\Big]\leq \Ex_{\mu_V^N}\Big[\lambda_{\max}^N(\mathbf{A})^{|\bm{w}|}\Big]\leq \kappa^{|\bm{w}|}\Gamma\big(\tfrac{|\bm{w}|}{2}+1\big),
    \end{equation}
    for some $\kappa>0$, where $\Gamma$ denotes the usual Gamma function. The second inequality uses a common equivalent definition of sub-Gaussian random variables, see for example \cite[Proposition~2.5.2]{Vershynin}. Then, by the factorial decay of the signature
    \begin{align*}
        \sum_{\bm{w}\in\WW_d}\Ex_{\mu_V^N}\Big[\big|\tr{A_{\bm{w}}^N}\big|\Big]\big|\SS_{s,t}^{\bm{w}}(\gamma)\big|&\leq \sum_{\bm{w}\in\WW_d}\kappa^{|\bm{w}|}\Gamma\big(\tfrac{|\bm{w}|}{2}+1\big)\big|\SS_{s,t}^{\bm{w}}(\gamma)\big|
        \leq \sum_{n=0}^\infty \frac{\tilde{\kappa}^{n}\Gamma\big(\tfrac{n}{2}+1\big)}{\Gamma(n+1)}
        <\infty,
    \end{align*}
    for some $\tilde{\kappa}>0$. Then Fubini-Tonelli and the dominated convergence theorem allows for the interchange of limits and expectation with the sum of words.
\end{proof}
We now turn to the main result of this section.
\begin{theorem}\label{thm: integro_diff_eq}
    The functional $\pdg{s}{t}$ satisfies the path-dependent integro-differential equation
    \begin{equation}\label{eq: integro_diff_eq} 
    \pdg{s}{t}=1-\int\limits_{s\leq u\leq v\leq t}\pdg{s}{u}\pdg{u}{v}\langle\dif\gamma_u,\dif\gamma_v\rangle - \sum_{k=1}^d\int_s^t\DD_W^k\pdg{s}{u}\dif\gamma_u^k.
    \end{equation}
    with the boundary condition $\pdg{s}{s}=1$ for all $s\in [0,T]$.
\end{theorem}
\begin{proof}
    Recall from \cref{thm: S-D_uniqueness} that
    \[
    \lim_{N\to\infty}\mathbb{E}_{\mu_V^N}\big[\tr{A_{\bm{w}}}\big]=\tau_V(X_{\bm{w}}),
    \]
    where $\tau_V$ is the unique non-commutative law associated with $V$ solving the corresponding Schwinger-Dyson equation. Consider now the following expansion, where we use~\eqref{eq: free_diff_monomial} for the first equality.
    \begin{align*}
    \sum_{\bm{w}\in\WW_d}i^{|\bm{w}|+1}\tau_V&\otimes\tau_V(\del_k X_{\bm{w}})\SS_{s,t}^{\bm{w}}(\gamma)=\sum_{\bm{u},\bm{v}\in\WW_d}i^{|\bm{u}|+|\bm{v}|+2}\tau_V(X_{\bm{u}})\tau_V(X_{\bm{v}})\SS_{s,t}^{\bm{u}k\bm{v}}(\gamma)\\
    &=-\sum_{\bm{u},\bm{v}\in\WW_d}i^{|\bm{u}|+|\bm{v}|}\tau_V(X_{\bm{u}})\tau_V(X_{\bm{v}})\int_s^t\SS_{s,u}^{\bm{u}}(\gamma)\SS_{u,t}^{\bm{v}}(\gamma)\dif\gamma_u^k\\
    &=-\int_s^t\left(\sum_{\bm{u}\in\WW_d}i^{|\bm{u}|}\tau_V(X_{\bm{u}})\SS_{s,u}^{\bm{u}}(\gamma)\right)\left(\sum_{\bm{v}\in\WW_d}i^{|\bm{v}|}\tau_V(X_{\bm{v}})(\gamma)\SS_{u,t}^{\bm{v}}\right)\dif\gamma_u^k\\
    &=-\int_s^t\pdg{s}{u}\pdg{u}{t}\dif\gamma_u^k.
\end{align*}
We now expand the left-hand-side using the Schwinger-Dyson equation
\begin{align*}
    \sum_{\bm{w}\in\WW_d}i^{|\bm{w}|+1}\tau_V&\otimes\tau_V(\del_k X_{\bm{w}})\SS_{s,t}^{\bm{w}}(\gamma)=\sum_{\bm{w}\in\WW_d}i^{|\bm{w}|+1}\tau_V(X_{\bm{w}} D_kV)\SS_{s,t}^{\bm{w}}(\gamma)\\
    &=\sum_{\bm{w}\in\WW_d}i^{|\bm{w}|+1}\bigg(\tau(X_{\bm{w}}X_k)+\sum_{j=1}^mg_j\tau_V(X_{\bm{w}} X_{D_k\bm{v}_j})\bigg)\SS_{s,t}^{\bm{w}}(\gamma)\\
    &=\sum_{\bm{w}\in\WW_d}i^{|\bm{w}|}\tau(X_{\bm{w}})\nabla^k\SS^{\bm{w}}_{s,t}(\gamma)+\sum_{j=1}^m(-i)^{|\bm{v}_j|-2}g_j\sum_{\bm{w}\in\WW_d}\tau_V(X_{\bm{w}})i^{|\bm{w}|}\nabla^{D_k\bm{v}_j}\SS_{s,t}^{\bm{w}}(\gamma)\\
    &=\nabla^k\sum_{\bm{w}\in\WW_d}i^{|\bm{w}|}\tau(X_{\bm{w}})\SS^{\bm{w}}_{s,t}(\gamma)+\sum_{j=1}^m(-i)^{|\bm{v}_j|-2}g_j\nabla^{D_k\bm{v}_j}\sum_{\bm{w}\in\WW_d}\tau_V(X_{\bm{w}})\SS_{s,t}^{\bm{w}}(\gamma)\\
    &=\nabla^k\pdg{s}{t}+\sum_{j=1}^m(-i)^{|\bm{v}_j|-2}g_j\nabla^{D_k\bm{v}_j}\pdg{s}{t}\\
    &=\big(\nabla^k+\DD_W^k\big)\pdg{s}{t}.
\end{align*}
Equating these identities and integrating both sides with respect to $\dif\gamma^k$ and then summing over $k$ results in~\eqref{eq: integro_diff_eq}:
\begin{equation}
    \pdg{s}{t}=1-\int\limits_{s\leq u\leq v\leq t}\pdg{s}{u}\pdg{u}{v}\langle\dif\gamma_u,\dif\gamma_v\rangle - \sum_{k=1}^d\int_s^t\DD_W^k\pdg{s}{u}\dif\gamma_u^k.
\end{equation}
\end{proof}
\begin{example}[Straight Lines]
    Fix $T=1$ and consider the case of a straight line $\gamma =tv$, for $v\in\R^d$. We note that $\dif\gamma_t^k=v^k\dif t$ and that $\pdg{s}{t}$ is now a function of only $(t-s)$, so that~\eqref{eq: integro_diff_eq} may be written as
    \[
    \DD_V \langle\gamma\rangle_t=-\int_{0}^t \langle\gamma\rangle_{u}\langle\gamma\rangle_{t-u}v^k\dif t,
    \]
    for $t\in [0,T]$.
We thus recover the classical result~\eqref{eq:matrixmodelloop}. 
\end{example}
\begin{example}[GUE]
    Consider the case of the quadratic potential $V=\tfrac{1}{2}\sum_{i=1}^dX_i^2$, for which $\mu_V^N$ corresponds to the GUE. Then each $\DD_V^k$ is given by $\Delta_k$. We may then integrate both sides of the equation
    \[
    \DD_V^k\pdg{s}{t}=-\int_s^t\pdg{s}{u}\pdg{u}{t}\dif\gamma_u^k,
    \]
    with respect to $\gamma^k$ to obtain
    \[
        \sum_{\bm{w}\in\WW_d}i^{|\bm{w}|+1}\tau_V(X_{\bm{w}k})\SS^{\bm{w}k}_{s,t}(\gamma)=\int_s^t\DD_V^k\pdg{s}{r}\dif\gamma_r^k=-\int\limits_{s<u<r<t}\pdg{s}{u}\pdg{u}{r}\dif\gamma_u^k\dif\gamma_r^k
    \]
    Summing over $k$, we see that
    \begin{align*}
    -\int\limits_{s<u<r<t}\pdg{s}{u}\pdg{u}{r}\langle\dif\gamma_u,\dif\gamma_r\rangle&=\sum_{k=1}^d\int\limits_{s<u<r<t}\pdg{s}{u}\pdg{u}{r}\dif\gamma_u^k\dif\gamma_r^k\\
    &=-\sum_{k=1}^d\sum_{\bm{w}\in\WW_d}i^{|\bm{w}|+1}\tau_V(X_{\bm{w}k})\SS^{\bm{w}k}_{s,t}(\gamma)\\
    &=-1+\sum_{\bm{w}\in\WW_d}i^{|\bm{w}|}\tau_V(X_{\bm{w}})\SS^{\bm{w}}_{s,t}(\gamma)\\
    &=\pdg{s}{t} - 1,
    \end{align*}
    which is the equation obtained in \cite{CT_free}.
\end{example}

To conclude this subsection, we consider the uniqueness of solutions to~\eqref{eq: integro_diff_eq}. A complete theory for uniqueness in a general space of functions on path space is a topic for future research. Instead, we consider infinite linear functionals on the signatures.
\begin{proposition}
    Consider the set of infinite linear functionals, $\ell:\XX\to\C$, on the signature,
    \[
    \ell(\gamma_{s,t})= \sum_{\bm{w}\in\WW_d}i^{|\bm{w}|}a_{\bm{w}}\SS_{s,t}^{\bm{w}}(\gamma),
    \]
    for which $|a_{\bm{w}}|\leq C^{|\bm{w}|}$ for some $C>0$. Suppose that
    \begin{equation}\label{eq: lf_sig}
    \ell(\gamma_{s,t})=1-\int\limits_{s\leq u\leq v\leq t} \ell(\gamma_{s,u}) \ell(\gamma_{u,v})\langle\dif\gamma_u,\dif\gamma_v\rangle - \sum_{k=1}^d\int_s^t\DD_W^k \ell(\gamma_{s,u})\dif\gamma_u^k,
    \end{equation}
    holds for all $\gamma\in\XX$, then $a_{\bm{w}}=\tau_V(X_{\bm{w}})$ for all $\bm{w}\in\WW_d$.
\end{proposition}
\begin{proof}
    Expanding~\eqref{eq: lf_sig} gives
    \begin{align*}
        \sum_{\bm{w}\in\WW_d}i^{|\bm{w}|}a_{\bm{w}}\SS_{s,t}^{\bm{w}}(\gamma) &= 1+\sum_{k=1}^d\sum_{\bm{u},\bm{v}\in\WW_d}i^{|\bm{u}|+|\bm{v}|+2}a_{\bm{u}}a_{\bm{v}}\SS_{s,t}^{\bm{u}k\bm{v}k}(\gamma)\\
        &\qquad +\sum_{k=1}^d\sum_{\bm{u}\in\WW_d}\sum_{j=1}^mi^{|\bm{u}|+|\bm{v}_j|-2 - |\bm{v}_j|}g_ja_{\bm{u}D_k\bm{v}_j}\SS_{s,t}^{\bm{u}k}(\gamma)
    \end{align*}
    Noting that the left-hand-side may be expanded as
    \[
    \sum_{\bm{w}\in\WW_d}i^{|\bm{w}|}a_{\bm{w}}\SS_{s,t}^{\bm{w}}(\gamma) = 1+\sum_{k=1}^d\sum_{\bm{w}\in\WW_d}i^{|\bm{w}|+1}a_{\bm{w}k}\SS_{s,t}^{\bm{w}k}(\gamma),
    \]
    means that we may equate the coefficients of $\SS^{\bm{w}k}$ to obtain the system of equations
    \[
    a_{\bm{w}k} = \sum_{\bm{w}=\bm{u}k\bm{v}}a_{\bm{u}}a_{\bm{v}}-\sum_{j=1}^mg_ja_{\bm{w}D_kv_j},
    \]
    which after rearranging yields
    \[
     \sum_{\bm{w}=\bm{u}k\bm{v}}a_{\bm{u}}a_{\bm{v}} =a_{\bm{w}k}+\sum_{j=1}^mg_ja_{\bm{w}D_kv_j},
    \]
    which implies that the extension of $\ell$ to a linear functional $\tau_a:\CC_d\to\C$ satisfies the Schwinger-Dyson equation $\tau_a\otimes\tau_a(\del_kX_{\bm{w}})=\tau_a(X_{\bm{w}} D_kV)$. By the uniqueness of solutions to the Schwinger-Dyson equation, \cref{thm: S-D_uniqueness}, we must have $\tau_a=\tau_V$.
\end{proof}

\section{Quantum algorithm}\label{sec:quantumsignature}

In this section we introduce quantum analogues of the classical path signature and signature kernel. We will study three main objects. Firstly, the \textit{quantum path development}--this is a random quantum circuit that we denote by $U_{\gamma}^Q(\alpha(m), n, K)$. The circuit acts on $n$ qubits, dependent on $m$ random Pauli strings labelled by~$\alpha$, an approximation quality $K$, and a path $\gamma$. We will show (\cref{thm:sparse_approximation}) that the expected value of this circuit converges to the Gaussian unitary path developments studied in the previous section. Secondly, we define the \textit{quantum path signature} $\mathcal{S}^{Q}_{\epsilon}$ to be the quantum feature map associated to the quantum circuit. Finally, in \cref{subsec:quantumalg}, we discuss a \textit{quantum signature kernel}, denoted $k^{Q}(\sigma,\tau)$, which is realised as the output of a quantum algorithm $\mathcal{A}$ built using the quantum path development $U_{\gamma}^Q(\alpha(m), n, K)$. One of the main results of this section is \cref{thm:quantumalg} demonstrating the efficiency of this algorithm. Along the way, we give theoretical guarantees for an analogous classical algorithm (\cref{thm:classicalalg}) to compute the randomised path development.

\subsection{Preliminaries}
In the following, we review basic definitions and concepts from quantum computing. Whilst this material will likely be familiar to readers with a background in quantum information theory, we include it to ensure the work remains accessible to researchers from the path signature and machine learning communities, in line with our aim of connecting ideas across these fields. For a more comprehensive introduction, we refer the reader to the standard text \cite{nielsen2010quantum}.

\paragraph{Quantum states.} 
We consider the Hilbert space of $n$ qubits
\begin{equation}
    \mathcal{H} = \underbrace{\mathbb{C}^2  \otimes \cdots \otimes \mathbb{C}^2}_{n\text{-times}} ,
\end{equation}
which has dimension $N=2^n$. We write $|\psi \rangle$ for vectors, henceforth referred to as (pure) quantum states, in $\mathcal{H}$. We may write the basis of tensor product states of $\mathcal{H}$ as $|x \rangle = |x_1 \rangle \otimes \cdots \otimes |x_n \rangle$ where each $x_{i} \in \{0,1\}$. Each state may then be identified with a binary string $x \in \{0,1\}^n$ which defines an isomorphism with $\mathbb{C}^{N=2^n}$. We define $|\mathbf{0} \rangle := |0\rangle \otimes \cdots \otimes |0\rangle \in \mathcal{H}$. 

More general quantum states, referred to as \emph{mixed states}, are represented by a \emph{density operator} $\rho$ acting on $\mathcal{H}$, which is a positive semi-definite Hermitian operator satisfying $\Tr{\rho} = 1$. If, in addition, $\Tr{\rho^2} = 1$, then $\rho$ corresponds to a pure state and can be written in the form $\rho = |\psi\rangle \langle \psi|$, considered as an element of $\mathcal{H}\otimes\mathcal{H}^*$, for some state $|\psi\rangle \in \mathcal{H}$. In the general case,  $\rho$ describes a classical probabilistic mixture of pure states and can be expressed as a convex combination:
\begin{equation}
\rho = \sum_{i=1}^k p_i |\psi_i\rangle \langle \psi_i|,
\end{equation}
where $p_i \geq 0$, $\sum_i p_i = 1$, and each $|\psi_i\rangle \in \mathcal{H}$ is a pure state. We denote this (Hilbert) space of general quantum states by $S(\mathcal{H})$. A quantum circuit is a unitary matrix $U$ acting on $\mathcal{H} = (\mathbb{C}^2)^{\otimes n}$ built from a finite sequence of elementary unitaries chosen from a \textit{universal gate set} \cite{kitaev1997quantum}.

\paragraph{Measurements.}
In the present context, the measurement of a state $\rho$ will be described by a set of orthogonal projectors $\{\Pi_i\}_{i=1}^k$ on $\mathcal{H}$ satisfying $\sum_{i=1}^k \Pi_i = I_{N}$ and $\Pi_i \Pi_j = \delta_{ij} \Pi_i$. The measurement describes the possible outcomes of a discrete random variable~$\Pi$ taking values in $\{1,\ldots,k\}$, with associated probabilities 
\begin{equation}\label{eq:Born_rule}
    \mathbb{P}(\Pi = i) = \tr{\rho \Pi_i}.
\end{equation}
This is known as Born's rule.

\paragraph{Quantum feature map.}
Let $\mathcal{X}$ be a set of data;
a quantum feature map \cite{schuld2019quantum}is an embedding
\begin{equation}\label{eq:qfm}
    \rho: \mathcal{X} \to S(\mathcal{H}),
\end{equation}
where the feature space is the space of quantum states on $\mathcal{H}$, i.e. we map~$x\in\mathcal{X}$ to a density operator~$\rho(x)$.

\subsection{Quantum circuit}\label{subsec:quantumcircuit}
We now turn to a quantum circuit implementation of the unitary path development. The Hilbert space of qubits $\mathcal{H}$ has dimension $N=2^n$. In the geometrical setup outlined in \cref{subsec:Wilsonlines}, this Hilbert space now plays the role of the associated $G=U_N$ bundle in the fundamental representation and the quantum circuit will play the role of the parallel transport operator between representation vectors.

We assume throughout this section that paths $\gamma$ are continuous piecewise linear $\gamma:[0,T] \to \mathbb{R}^d$. It is natural to consider this class of paths as any smooth path can be approximated in an appropriate topology by sequences of piecewise linear paths. The time interval $[0,T]$ is divided into intervals $0 = t_0 < t_1 < \ldots < t_L = T$ and we write
\begin{equation}\label{eq:piecewisepath}
    \gamma = \gamma_{1} \star \cdots \star \gamma_{L},
\end{equation}
where each $\gamma_l$ is a linear path on $[t_{l-1},t_{l}]$. The data specifying the path is then $L$ vector increments, denoted $\Delta_{l}^{\nu} := \gamma^{\nu}(t_{l}) - \gamma^{\nu}(t_{l-1})$ for $l=1,\ldots,L$ and $\nu=1,\ldots,d$. For piecewise linear paths, the length is equal to the total one-variation which we denote by $\Delta_{\gamma} = \sum_{l=1}^L \|\Delta_l\|$.

As we will see more precisely later, the computational difficulty of computing the path development is controlled by the $U_N$ dimension $N=2^n$, the total one-variation $\Delta_{\gamma}$ (\textit{i.e.} the length of the path) and the quality of the approximation $\epsilon$ arising due to the discretisation as a quantum circuit.

\begin{definition}[Pauli strings]\label{def:PS}
A set of Hermitian operators on $\mathbb{C}^2$ given by
\begin{equation}\label{eq:paulis}
    \sigma_I = \begin{pmatrix} 1 & 0 \\ 0 & 1 \end{pmatrix}, \quad
    \sigma_X = \begin{pmatrix} 0 & 1 \\ 1 & 0 \end{pmatrix}, \quad
    \sigma_Y = \begin{pmatrix} 0 & -i \\ i & 0 \end{pmatrix}, \quad
    \sigma_Z = \begin{pmatrix} 1 & 0 \\ 0 & -1 \end{pmatrix}.
\end{equation}
A (real) linear basis for $\mathfrak{u}_N$ is then given by tensor products of the Pauli matrices
\begin{equation}
    \sigma_{\mathbf{w}} = \sigma_{w_1} \otimes \cdots \otimes \sigma_{w_n},
\end{equation}
where $\mathbf{w}=w_1\ldots w_n$ is a word with $n$ letters chosen from $\{I,X,Y,Z\}$ and $\otimes$ denotes the tensor product. The operator $\sigma_{\mathbf{w}}$ is referred to as a Pauli string. We write the set of words on Pauli matrices as $\PP$.
\end{definition}

We now consider the $N \times N$ unitaries $U_{\gamma}$ discussed in previous sections as operations on qubits. We will show that these may be decomposed into quantum gates. Let us first consider the $N$-dimensional path development for a fixed set of $d$ Hermitian matrices $A_1, \ldots A_d$, that is the solution to\footnote{In this equation and in the following we use the summation convention that repeated upper and lower indices are summed. E.g. $X_{\nu}X^{\nu} := \sum_{\nu=1}^d X_{\nu}X^{\nu}$.}
\begin{equation}\label{eq:piecewisePD}
    \dif U_{\gamma}(t) = i U_{\gamma}(t) A_{\nu} \dot{\gamma}^{\nu}_t \dif t, \quad U_{\gamma}(0) = I,
\end{equation}
at $t=T$. We denote the solution by $U_{\gamma}(A,N)$. We now turn to how to implement this unitary as a quantum circuit acting on the Hilbert space $\mathcal{H} = (\mathbb{C}^2)^{\otimes n}$.

\paragraph{Trotterisation.}
Let us discuss the solution to the path development equation~\eqref{eq:piecewisePD} for piecewise linear paths~\eqref{eq:piecewisepath} and for a fixed set of $N \times N$ matrices $\{A_{\nu}\}_{\nu=1}^d$. From the fact $U_{\gamma_1 \star \gamma_2} = U_{\gamma_1}U_{\gamma_2}$ we have that
\begin{equation}
    U_{\gamma}(A,N) = U_{\gamma_{1}}(A,N)U_{\gamma_{2}}(A,N)\cdots U_{\gamma_L}(A,N),
\end{equation}
where $U_{\gamma_l}(A,N)$ denotes the unitary time evolution operator/path development for the straight line segment on $[t_l,t_{l+1}]$. Each $U_{\gamma_l}(A,N)$ is then given by
\begin{equation}
    U_{\gamma_l}(A;N) = \exp( i \sum_{\nu=1}^d \Delta_l^{\nu}A_{\nu}),
\end{equation}
where $\Delta_l^{\nu} = \gamma^{\nu}(t_{l}) - \gamma^{\nu}(t_{l-1})$ are the (vector) increments of the path on each time step. Recall (Definition~\ref{def:PS}) that each Hermitian matrix $A_{\nu}$ may be expressed in the Pauli string basis. We write the expansion as
\begin{equation}\label{eq:Aexpansion}
    A_{\nu} = \sum_{\mathbf{w} \in P} \alpha_{\nu}^{\mathbf{w}}\sigma_{\mathbf{w}},
    \quad \nu=1,\ldots d,
\end{equation}
with $\{\alpha_{\nu}^{\mathbf{w}}\}_{\mathbf{w}}$ a set of $4^n$ real coefficients for each matrix and often write the development as $U_{\gamma}(\alpha,N)$ emphasising the dependence on these coefficients. We now consider the Trotterisation of this unitary to define a quantum circuit.

\begin{definition}
The first-order Trotterisation with $K$ subdivisions of each straight line interval $\gamma_l$, written $U_{\gamma_l}(\alpha,n,K)$, is defined by
\begin{equation}\label{eq:trotter}
    U_{\gamma_i}(\alpha;n,K) := \left[\prod_{\nu=1}^d \prod_{\mathbf{w} \in P} \exp(i \frac{\Delta_i^{\nu} \alpha_{\nu}^{\mathbf{w}}}{K}\sigma_{\mathbf{w}})\right]^K 
    = \left[\prod_{\nu=1}^d \prod_{\mathbf{w} \in P} P_{\mathbf{w}}\left(\frac{\Delta_i^{\nu} \alpha_{\nu}^{\mathbf{w}}}{K}\right)\right]^K,
\end{equation}
where we write the shorthand $\alpha$ for the set of coefficients $\{\alpha_{\nu}^{\mathbf{w}}\}$ that specify~$A$.  
In the second equality, we have re-written the exponentials as Pauli rotations $P_{\mathbf{w}}(\theta) := \exp(i \theta \sigma_{\mathbf{w}})$ with small angles. Finally, we define the Trotterised path development of~$\gamma$ by
\begin{equation}\label{eq:densecircuit}
    U_{\gamma}(\alpha;n,K) := U_{\gamma_L}(\alpha;n,K)U_{\gamma_{L-1}}(\alpha;n,K)\cdots U_{\gamma_1}(\alpha;n,K).
\end{equation} 
\end{definition}
We have thus constructed a quantum circuit from repeated application of multi-qubit Pauli rotations with small angles that approximates the path development with a fixed set of Hermitian matrices determined by coefficients $\{\alpha_\nu^{\mathbf{w}}\}$. Intuitively, one may view the Trotterisation as dividing each interval $[t_i,t_{i+1}]$ further into~$K$ subdivisions and one recovers the full development in the limit $K \to \infty$. The following lemma quantifies the quality of this approximation.

\begin{lemma}\label{lemma:trotter}
We consider the Trotterisation into $K$ subdivisions of the path development $U_{\gamma}$ for a piecewise linear path~$\gamma$ defined by a fixed gauge field specified by coefficients $\{\alpha_{\nu}^{\mathbf{w}}\,:\, |\alpha_{\nu}^{\mathbf{w}}|\le 1/m\}$ in the Pauli basis. Then
\begin{equation}
    \|U_{\gamma}(\alpha,n,K) - U_{\gamma}(\alpha,n)\| \le \frac{|P|_{\alpha}}{K} \Delta_{\gamma}^2 + \OO\left(\frac{1}{K^2}\right),
\end{equation}
where $|P|_{\alpha}$ denotes the maximum number (over $\nu$) of non-zero coefficients in the Pauli string expansion and $\Delta_{\gamma}$ denotes the total one-variation of the path $\gamma$. 
\end{lemma}
\begin{remark}
When $\gamma(t)=t$ is a one-dimensional linear path segment on  $[0,t]$, the lemma recovers the familiar Trotter error bound from Hamiltonian simulation theory that is quadratic $\OO(t^2)$ in the time interval (length of the path) and $\OO(1/K)$ in the Trotter subdivision.
\end{remark}

\begin{proof}
    First, recall the standard Trotterisation error bound \cite{suzuki1976generalized}. Suppose a Hamiltonian is given as a sum of bounded operators $H = \sum_{j=1}^N H_j$ then we have, in operator norm,
    \begin{equation}\label{eq:standardTrotter}
        \left\|\left(\prod_{j=1}^N \exp\left\{\frac{i H_j}{K}\right\}\right)^K - e^{iH}\right\| \le \frac{1}{2K} \sum_{i<j}\|[H_i,H_j]\| + \OO\left(\frac{1}{K^2}\right).
    \end{equation}
    Consider now the Hamiltonian $H = \sum_{\nu=1}^d \sum_{\mathbf{w} \in P} H_{\nu,\mathbf{w}}$ with $H_{\nu,\mathbf{w}} = \Delta_i^\nu \alpha_\nu^{\mathbf{w}}\sigma_{\mathbf{w}}$. 
    From the Lie algebra structure of the Pauli matrices, then
    $\|[\alpha_{\nu}^\mathbf{w}\sigma_{\mathbf{w}}, \alpha_{\nu}^\mathbf{w'}\sigma_{\mathbf{w'}}]\| \le 2$, and~\eqref{eq:standardTrotter} yields
    \begin{equation}
        \|U_{\gamma_i}(\alpha;n,K) - U_{\gamma_i}(\alpha;n)\| \le \frac{|P|_{\alpha_{\nu}}}{K} \sum_{\nu=1}^d|\Delta_i^{\nu}|^2 + \OO\left(\frac{1}{K^2}\right).
    \end{equation}
    where $|P|_{\alpha_{\nu}}$ denotes the number of $\mathbf{w}$ for which  $\alpha_{\nu}^{\mathbf{w}}$ is non-zero. Now we would like to consider a bound on 
    \begin{equation}
        \|U_{\gamma}(\alpha,n,K) - U_{\gamma}(\alpha,n)\|
        = \left\| \prod_{i=1}^L U_{\gamma_i}(\alpha,n,K) - \prod_{i=1}^LU_{\gamma_i}(\alpha,n)\right\|.
    \end{equation}
    Telescoping the product and using the fact that the operators are unitary we simply obtain the sum of the individual bounds between the product components
    \begin{equation}
    \begin{split}
        \|U_{\gamma}(\alpha,n,K) - U_{\gamma}(\alpha,n)\| &\le \frac{|P|_{\alpha}}{K}\sum_{i=1}^L \sum_{\nu=1}^d |\Delta_i^{\nu}|^2 + \OO\left(\frac{1}{K^2}\right)\\
        &= \frac{|P|_{\alpha}}{K} \Delta_{\gamma,2}^2 + \OO\left(\frac{1}{K^2}\right), \\
        & \le \frac{|P|_{\alpha}}{K} \Delta_{\gamma}^2 + \OO\left(\frac{1}{K^2}\right),
    \end{split}
    \end{equation}
    where in the second line $\Delta_{\gamma,2}^2$ denotes the square of the total two-variation of the path and in the third line we have used that for piecewise linear paths the two-variation is bounded by the one-variation. In the above $|P|_{\alpha} := \max_{\nu} |P|_{\alpha_{\nu}}$.
\end{proof}

\begin{remark}
The circuit $U_{\gamma}(\alpha,n,K)$ may be further decomposed into a fundamental gate set, for example using the methods described in \cite{sriluckshmy2023optimal}. 
\end{remark}

\paragraph{Randomisation and sparsification.}
We now discuss how to randomise over the Hermitian matrices $A_{\nu}$ (or equivalently the coefficients $\alpha$ in the Pauli string expansion of $A$). We sample $A_\nu$ from the $N=2^n$ dimensional Gaussian unitary ensemble (GUE) by sampling coefficients of the Pauli strings in the circuit~\eqref{eq:densecircuit}. 
In recent work \cite{chen2024sparse}, the authors define a sparse ensemble of random Pauli strings.
They also show that ground states of the associated Hamiltonian are efficient to prepare. In the following, we use the same matrix ensemble in a different context that allows us to consider a sparsification of Pauli strings in the path development circuit~\eqref{eq:densecircuit}.
We first recall the definition of the random Pauli ensemble and then argue that the first moment of this random Pauli ensemble converges to the semi-circular law with a sufficiently fast rate.

\begin{definition}{(Random Pauli ensemble)}\label{def:SPS}
    Let $N=2^n$.
    A matrix $A$ belongs to the random Pauli ensemble with $m$ Pauli strings if
    \[
    A=\frac{1}{\sqrt{m}}\sum_{i=1}^mr^i\sigma^{i},
    \]
    where $r_i\stackrel{\text{i.i.d.}}{\sim}\mathrm{Rademacher}(\nicefrac{1}{2})$ and $\sigma^{i}=\sigma(1)^i\otimes\cdots\otimes\sigma(n)^{i}$, where $\sigma(j)^i\stackrel{\text{i.i.d.}}{\sim}\PP$, where $\sigma\sim \PP$ denotes a random variable uniformly distributed on the set of Pauli matrices $\PP$.  
\end{definition}

The distribution of $\nu =1,\ldots, d$ matrices from the random Pauli ensemble of Definition~\ref{def:SPS}, may be equivalently described by a random variable $\alpha$ for the coefficients in the expansion~\eqref{eq:Aexpansion} where for each $\nu = 1,\ldots, d$,
the coefficient $\alpha_{\mathbf{w}}^{\nu}$ is non-zero for only $m$ uniformly chosen $\mathbf{w}$ where it takes the two values $\{-1/\sqrt{m},1/\sqrt{m}\}$ with equal probability. We henceforth denote this random variable by $\alpha(m) = \alpha(m)_{\mathbf{w}}^{\nu}$. Accordingly, we write the path development associated to $\alpha(m)$ as
\begin{equation}\label{eq:SPdev}
    \dif U_{\gamma_{s,t}}^{\text{SP}}(\alpha(m),n) = i U_{\gamma_{s,t}}^{\text{SP}}(\alpha(m),n) A_{\nu}^{m} \dot{\gamma}^{\nu}_t \dif t,\quad U_{\gamma_{s,s}}^{\text{SP}}(\alpha(m),n) = I_{2^n},
\end{equation}
where the matrices $A_{\nu}^{m}$ have expansions $A_{\nu}^{m} = \sum_{\mathbf{w} \in P} \alpha(m)_{\nu}^{\mathbf{w}}\sigma_{\mathbf{w}}$ with $m$ non-zero terms. 
We write the solution as $U_{\gamma_{s,t}}^{\text{SP}}(\alpha(m),n)$.

Next we present the main result of this subsection, which states that the path development driven by random Pauli ensemble matrices approximates well the unitary path development. 
The proof is given in Appendix~\ref{proof_sparse_approx} and relies crucially on the following proposition about the moments of the Pauli ensemble. 
For completeness, and because it may be of independent interest, we restate this proposition below.
\begin{proposition}
    For each $n$, let $A_1^n,\ldots,A_d^n$ be i.i.d. matrices from the random Pauli string ensemble of dimension $N=2^n$ with $m_n$ strings each.
    Then
    \[
\left|\Ex\Big[\tr{A_{\bm{w}}^n}\Big] - \tau(X_{\bm{w}})\right| \leq (4e)^pp!\left(m_n^{-1} + 2^{-2n}\right),
    \]
    where $\tau$ is the law of $d$-free semicircular random variables and $\bm{w}\in\WW_d$ has length $2p$ for some $p\in\N$. In the above $A_{\bm{w}}$ denotes the product $A_{w_1}\ldots A_{w_{2p}}$.
\end{proposition}

\begin{theorem}\label{thm:sparse_approximation}
    Fix $n>0$ and let $\alpha(m)$ be the random variable defined above. 
    Let
    \[
    \pdg{s}{t}=\sum_{\bm{w}\in\WW_d}i^{|\bm{w}|}\tau(X_{\bm{w}})\SS^{\bm{w}}_{s,t}(\gamma),
    \]
    be the GUE randomised path development, then there exists a constant $C>0$ depending only on the length of $\gamma$ and the dimension $d$ of the path target space, for which
    \begin{equation}
        \Ex\bigg[\Big(\tr{U^{\text{SP}}_{\gamma_{s,t}}(\alpha(m),n)}-\pdg{s}{t}\Big)^2\bigg]< C\big(m^{-1}+4^{-n}\big).
    \end{equation}
\end{theorem}
We now combine the previous arguments to define a quantum path development. We consider a piecewise linear path $\gamma: [a,b] \to \mathbb{R}^d$ given as $L$ specified increment vectors \( \{\Delta_l^1, \Delta_l^2, \ldots, \Delta_l^d\}_{l=1}^L \). 
We define the \textit{quantum path development} which combines (1) the sparse Pauli approximation of the large $N$ GUE development of \Cref{thm:sparse_approximation} and (2) the Trotterisation approximation of the sparse Pauli development of \Cref{lemma:trotter}.

\begin{definition}[Quantum path development]\label{def:quantumcircuit}
The \textit{quantum path development} is the random quantum circuit $U^{Q}_{\gamma}(\alpha(m),n,K)$ defined in terms of Pauli rotations, following~\eqref{eq:densecircuit}:
\begin{equation}
    U^Q_{\gamma}(\alpha(m),n,K) = \prod_{l=1}^L \left[\prod_{\nu=1}^d \prod_{i=1}^m P_{\mathbf{w_{i,\nu}}}(\Delta_l^{\nu} \alpha_{\nu}^{\mathbf{w}_i}/K)\right]^K,
\end{equation}
where $\mathbf{w}_{i,\nu}$ denotes the $i=1,\ldots,m$ non-zero coefficients of $\alpha_{\nu}^{\mathbf{w}}$ for each $\nu$. The circuit depends on parameters $n$ (number of qubits), Trotter subdivisions $K$ and $m$ random Pauli strings. The circuit consists of $L\times K\times d\times m$ Pauli rotations.
\end{definition}

In the following section we demonstrate, as part of the proof of \cref{thm:quantumalg}, that the quantum path development well-approximates the unitary development of a path $\gamma$.

\paragraph{Quantum path signatures.}
We define the quantum path signature as a quantum feature map, in the sense discussed at the beginning of this section~\eqref{eq:qfm}, that allows one to embed a path $\gamma$ into the space of states on qubits. In the present case, $\mathcal{X}$ is a space of paths and we would like to embed the path $\gamma$ into the Hilbert space of a large number of qubits $n$, that is we would like to  consider a mean embedding of the form
\begin{equation}\label{eq:qps}
\int U_{\gamma}(A) | \mathbf{0} \rangle \langle \mathbf{0} | U_{\gamma}^{\dag}(A)\,\dif\mu_{V}^N(A) ,
\end{equation}
where $N$ is large and $\mu_{V}^N$ is the Gaussian measure. To make this well-defined, based on the previous discussion, we see that we may work with an approximate embedding into a truncated Hilbert space of $n$ qubits, with Trotterised time evolution into $K$ Trotter subdivisions and a sparse Pauli ensemble controlled by an integer $m$. 

\begin{definition}[Quantum path signature]\label{def:QPS}
We now define the quantum path signature as a quantum feature map $\mathcal{S}^{Q}(\gamma): \mathcal{X} \to S(\mathcal{H}_n)$ embedding paths $\gamma$ into a quantum state (density matrix) by the circuit of \cref{def:quantumcircuit}, namely
\begin{equation}
    \mathcal{S}^{Q}(\gamma) := \mathbb{E}_{\alpha(m)}\Big[U_{\gamma}^Q(\alpha(m),n,K) |\mathbf{0} \rangle \langle \mathbf{0}| U_{\gamma}^Q(\alpha(m),n,K)^{\dag}\Big],
\end{equation}
In applications, we propose that the parameters $n$, $K$ and $m$ should be chosen to scale as in \cref{thm:quantumalg}.
\end{definition}

\subsection{Quantum signature kernel}\label{subsec:quantumalg}

We now introduce a quantum algorithm, making use of the quantum path development circuit, to compute the trace of the GUE randomised path development\footnote{We drop the $\mu_{\infty}$ subscript compared with Definition~\eqref{eq:tracepathdevelopment} since we only consider here the large $N$ GUE path development.} 
\begin{equation}
    \langle \gamma \rangle := \lim_{N\to \infty}\mathbb{E}\left[\tr{U_{\gamma}(A,N)}\right].
\end{equation}
In the subsection that follows, we discuss how the algorithm may be further used to approximate the GUE kernel on path space $k^{\text{GUE}}:\mathcal{X} \times \mathcal{X} \to \mathbb{R}$ defined by
\begin{equation}
    k^{\text{GUE}}(\sigma,\tau) := \lim_{N\to \infty} \mathbb{E}\left[\tr{U_{\sigma}(A,N)U_{\tau}(A,N)^{\dag}}\right].
\end{equation}
as discussed in Section~\ref{subsec:sigkernels}. We define the \textit{quantum signature kernel} as the output of this quantum algorithm.

\begin{figure}
    \centering
    \[
    \begin{quantikz}
    \ket{0} \,& \gate{H} & \ctrl{1} & \gate{H} & \meter{} \\
    I/2^n \,& \qw & \gate{U^Q_{\gamma}(\alpha(m),n,K)} & \qw & \qw
    \end{quantikz}
    \]
    \caption{One clean qubit (DCQ1) circuit with the quantum path development.}
    \label{fig:one-clean-qubit}
\end{figure}
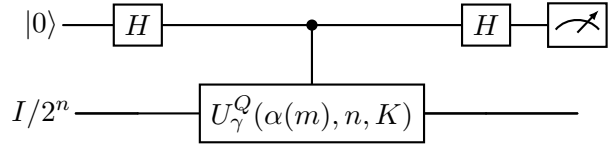

\paragraph{Problem statement.}
We consider a continuous piecewise linear path $\gamma$ on a partitioned time interval $0 = t_0 < t_1 < \cdots < t_L = T$ mapping into $\mathbb{R}^d$. The input data is thus the set of increments $\{\Delta_{l}^{\nu}\}$ with $l=1,\ldots,L$ and $\nu = 1,\ldots, d$. The computational problem is then to output a quantity~$Q$ that is $\epsilon$-additively close to the path development $\langle \gamma \rangle$ with small failure probability $\delta$.

\paragraph{Computation.}
We now define a quantum algorithm $\mathcal{A}(\{\Delta_{l}^{\nu}\},\epsilon,\delta)$ in the one-clean-qubit model of quantum computation~\cite{knill1998power} with random output~$Q$.
The algorithm takes as input the increments $\{\Delta_{l}^{\nu}\}$ together with an additive error~$\epsilon$ and a probability of failure~$\delta$.
We first take a sample~$\tilde{\alpha}(m)$ of the random variable~$\alpha(m)$ described below Definition~\ref{def:SPS}. 
We then embed the random quantum circuit $U^{Q}(\tilde{\alpha}(m), n, K)$ from Definition~\ref{def:quantumcircuit} as the controlled unitary in the one-clean-qubit circuit illustrated in \cref{fig:one-clean-qubit}. 
The initial state of the system is
\begin{equation}
    \rho_0 = |0\rangle \langle 0 | \otimes \frac{I}{2^n},
\end{equation}
and, after applying a Hadamard gate to the first qubit, a sample of the quantum path development circuit and an additional Hadamard gate, the state of the system is 
\begin{equation}
    \rho_1 = \frac{1}{2^{n+1}}\left[
    (|0\rangle\langle0| +|1\rangle\langle1|)\otimes I + |0\rangle\langle1| \otimes U_{\gamma}^Q(\tilde{\alpha}(m),n,K)^{\dag} + |1\rangle\langle0| \otimes U_{\gamma}^Q(\tilde{\alpha}(m),n,K)\right].
\end{equation}
Measuring the first qubit in the computational basis
then yields the probability of observing "$1$" by Born's rule~\eqref{eq:Born_rule},\footnote{Recall that the unitary path development is real valued.} 
\begin{equation}
\tr{|1\rangle \langle 1| \rho_1} = \frac{1}{2}
\left(1-\tr{U_{\gamma}^Q(\tilde{\alpha}(m),n,K)}\right).
\end{equation}
Finally, we define the output of the algorithm to be the estimator of $\tr{U_{\gamma}^Q(\tilde{\alpha}(m),n,K)}$. Namely
\begin{equation}
    Q_{\Delta}(M,m,n,K) := 1 - \frac{2 T_{\Delta}(M,m,n,K)}{M},
\end{equation}
where $T_\Delta(M,m,n,K)$ denotes the number of ones observed when initialising, running and measuring the circuit $M$ times. The output depends on the number of qubits $n$; the Trotter approximation $K$; the number of Pauli strings $m$; and the number of times $M$ to run the circuit of Figure~\ref{fig:one-clean-qubit}. \cref{thm:quantumalg} explains how to scale these parameters as functions of $\epsilon$ and $\delta$. The algorithm may be summarised as follows:\footnote{This algorithm is named \texttt{QSigKer} because of its direct application to computing a kernel in the following subsection.}

\begin{algorithm}[H]\label{alg:qalg}
\caption{\texttt{QSigKer} -- Quantum Signature Kernel Estimator $\mathcal{A}(\{\Delta_{l}^{\nu}\},\epsilon,\delta)$}
\begin{algorithmic}[1]
\Require Data of increments $\{\Delta_{l}^{\nu}\}$ and $\epsilon, \delta > 0$. Parameters: integers $m, n, K, M > 0$ for Pauli string number, qubit number, Trotter subdivisions and number of circuit runs respectively. These are selected according to \cref{thm:quantumalg}.
\Ensure Approximation $Q_\Delta(M,m,n,K)$ of $\langle \gamma \rangle$.
\State Set counter $T \gets 0$
\For{$i = 1$ to $M$}
    \State Sample $\tilde{\alpha}(m)$ from random variable $\alpha(m)$
    \State Initialise system in state $\rho_0 = |0\rangle\langle 0| \otimes \frac{I}{2^n}$
    \State Apply $H$ to first qubit. Apply unitary $U_{\gamma}^{Q}(\tilde{\alpha}(m), n, K)$ of definition~\ref{def:quantumcircuit}. Apply $H$ to first qubit. 
    \State Measure first qubit in computational basis obtaining outcome $x \in \{0,1\}$
    \If{$x = 1$}
        \State $T \gets T + 1$
    \EndIf
\EndFor
\State \Return $Q_\Delta(M,m,n,K) \gets 1 - \frac{2T}{M}$
\end{algorithmic}
\end{algorithm}

\begin{theorem}\label{thm:quantumalg}
    The positive integer parameters $m$, $n$, $K$ and $M$ 
    (Pauli string number, qubit number, Trotter subdivisions, number of circuit runs) may be chosen such that the output of $\mathcal{A}(\{\Delta_{l}^{\nu}\},\epsilon,\delta)$ estimates the unitary path development $\langle \gamma \rangle$ to within additive error $\epsilon$ with high probability $1-\delta$, namely such that
    \begin{equation}
        \mathbb{P}(|Q_{\Delta}(M,m,n,K) - \langle \gamma \rangle| \ge \epsilon) < \delta.
    \end{equation}
    The circuit requires $\text{log}(1/\epsilon,1/\delta)$ qubits and $\text{poly}(1/\epsilon,1/\delta)$ Pauli rotations.
\end{theorem}

\begin{proof}
See Appendix~\ref{subsec:quantumalgproof}
\end{proof}

\paragraph{Classical algorithm.}
Let us now define an analogous classical algorithm and perform the complexity analysis for comparison. The input data remains unchanged---the set of path increments \( \{\Delta_l^1, \Delta_l^2, \ldots, \Delta_l^d\}_{l=1}^L \). We define a probabilistic classical algorithm, \( \mathcal{A}^{\text{c}}(\{\Delta_i^{\nu}\},\epsilon,\delta) \), which, given the input data, an additive error $\epsilon$ and a probability of failure $\delta$, produces a random output that is $\epsilon$ close to the path development with high probability $1-\delta$. We follow a standard Monte Carlo approach, defined precisely as:\footnote{The truncated product is a poor practical approach to compute a matrix exponential, in particular it is numerically unstable (see for example \cite{moler2003nineteen} for a thorough discussion of numerical matrix exponentiation methods) but for our purpose of discussing separation between classical and quantum approaches it is sufficient.}

\begin{algorithm}[H]\label{alg:classical}
\caption{Classical Sampling $\mathcal{A}^{\text{c}}$}
\begin{algorithmic}
\Require Data: path increments $\{\Delta_l^1,\ldots,\Delta_l^d\}_{l=1}^L$ and $\epsilon,
\delta > 0$. Parameters: integer $N>0$, $M>0$ and $K>0$ are chosen dependent on $\epsilon$ and $\delta$ according to \cref{thm:classicalalg}.
\For{$m = 1$ to $M$}
\State Draw $d$ independent $N \times N$ GUE matrices $\{A_\nu\}_{\nu=1}^d$
    \For{$l = 1$ to $L$}
        \State Compute $A_l := \sum_{\nu=1}^d \Delta^{\nu}_lA_\nu$
        \State Compute $U_l^K := (1+A_l/K)^K$ using a truncated approximation dependent on $K>0$.
    \EndFor
    \State Compute $Q_m := \mathrm{Tr}(U_1^K U_2^K \cdots U_L^K)$
\EndFor
\State \Return $Q_{\Delta}(M,N,K) := \frac{1}{M} \sum_{m=1}^M Q_m$
\end{algorithmic}
\end{algorithm}

\begin{theorem}\label{thm:classicalalg}
    The output of $\mathcal{A}^{\text{c}}(\{\Delta_i^{\nu}\})$ satisfies
    \begin{equation}
        \mathbb{P}\left(\left|Q_{\Delta}(N,M,K)-\langle \gamma \rangle)\right|>\epsilon \right) < \delta,
    \end{equation}
    for small $\epsilon > 0$ provided that $M > 2 /\epsilon^2 \log(2/ \delta)$ samples are taken, the dimension $N$ of matrices satisfies $N > e^{2 \Delta_{\gamma}}/\epsilon^2$ and the matrix exponential is truncated to $K > \Delta_\gamma e^{2 \Delta_{\gamma}}/\epsilon$.
\end{theorem}
\begin{proof}
    See Appendix~\ref{subsec:classicalg}
\end{proof}

In \cref{thm:quantumalg} we argue that the resources required to approximate the development of a path with one-variation to accuracy $\epsilon$ with the quantum algorithm are polynomial in $1/\epsilon$ in time and logarithmic in $1/\epsilon$ in space.
In comparison, the classical algorithm outlined above requires $L \times M \times K$ matrix multiplications and hence, according to \cref{thm:classicalalg}, requires polynomial resources in $\epsilon$ in both time and space.

\paragraph{\texttt{QSigKer}.}
We now consider the computation of the GUE signature kernel using the quantum algorithm $\mathcal{A}$. In the notation of the previous section, the GUE kernel of two paths $\sigma$ and $\tau$ may be expressed as 
\begin{equation}
    k^{\text{GUE}}(\sigma,\tau) = \langle \gamma \rangle,
\end{equation}
where $\gamma$ denotes the concatenation $\gamma = \sigma \star \overleftarrow{\tau}$. This follows from the property that for unitary path developments we have $U_{\gamma} = U_{\sigma}U_{\tau}^{\dag}$. We conclude the work with the following definition:
\begin{definition}(\texttt{QSigKer})
The probabilistic quantum signature kernel $k_{\epsilon,\delta}^Q: \mathcal{X} \times \mathcal{X} \to \mathbb{R}$ is defined by
\begin{equation}
    k^{Q}_{\epsilon,\delta}(\sigma,\tau) := 1 - \frac{2 T_\Delta(M,m,n,K)}{M},
\end{equation}
where $T_\Delta(M,m,n,K)$ is the output of the algorithm of theorem \ref{thm:quantumalg}. In applications, the hyperparameters should be chosen to scale according to \cref{thm:quantumalg}.
\end{definition}

\section{Conclusion}
In this work, we have developed a novel bridge between the theory of path signatures and matrix models in quantum field theory, uncovering structural connections with implications for both machine learning and quantum computing. By interpreting randomised path developments through the lens of non-commutative probability and unitary matrix models, we derived a class of loop equations governing their evolution. This formalism naturally led us to define the quantum path signature and the quantum signature kernel as quantum analogues of their classical counterparts, realised via quantum circuits acting on qubits. In particular, we showed that in the Gaussian case, these constructions admit an efficient quantum implementation using random ensembles of Pauli strings. This provides a concrete quantum algorithm for approximating signature kernels, with potential applications in quantum-enhanced time series analysis.

\paragraph{Future directions.}
This work represents an initial step toward importing tools from theoretical physics and quantum computing into the study of path signatures in machine learning. As such, we believe our work opens several avenues for future research, including the exploration of other matrix model potentials, deeper study of universality phenomena in quantum path signatures, and the practical integration of these quantum kernels into hybrid classical-quantum machine learning pipelines. 

On the matrix model side, the theory of \emph{topological recursion} \cite{eynard2007invariants} offers a powerful framework for computing \(1/N\) corrections to matrix model observables. A natural extension would be to explore the ``loop'' version of our model by applying topological recursion to correlators of the form
\begin{equation}
\begin{split}
    W_n\left(\gamma_1,\ldots,\gamma_n;\{g_k\}\right) &= \left\langle \mathrm{tr}\left(\textrm{P}\exp \int_{\gamma_1} A \right) \cdots \mathrm{tr}\left(\textrm{P}\exp \int_{\gamma_n} A \right) \right\rangle \\
    &=\sum_{g\ge0}N^{2-2g-n} W_{g,n}\left(\gamma_1,\ldots,\gamma_n;\{g_k\}\right),
\end{split}    
\end{equation}
where the insertions correspond to a collection of paths \(\gamma_1, \ldots, \gamma_n\). This approach would enable a systematic investigation of finite-\(N\) corrections to unitary path developments and their induced kernels.

On the machine learning side, many open questions remain concerning the statistical and computational properties of the interacting path kernels introduced in Section~\ref{sec:unitarydevelopments}. In particular, it would be valuable to establish theoretical guarantees such as generalisation bounds or sample complexity estimates. Extending these results to the quantum path signature defined in Definition~\ref{def:QPS} is especially compelling, as it could provide a rigorous foundation for the use of quantum path features in downstream learning tasks.

\printbibliography[heading=bibintoc,title={References}]
\clearpage
\appendix
\section{Technical proofs}\label{app: proofs}
The following lemma may be derived by keeping track of the constants in the various equivalent characterisations of sub-Gaussian random variables, as in \cite[Proposition~2.5.2]{Vershynin} for example. We include a self-contained proof for completeness.
\begin{lemma}\label{lem: sub_gaussian_equiv}
    Let $(X_N)_{N\geq 1}$ be a sequence of real-valued random variables. Then the following are equivalent
    \begin{enumerate}[label=(\arabic*)]
        \item There exist $\alpha>0$ and $t_0<\infty$ such that for all $t\geq t_0$ and $N\geq 1$,
        $        \P(X_N\geq t)\leq e^{-\alpha Nt^2}$;
        \item There exists $C>0$ such that for all $s$ sufficiently small and all $N\geq 1$, then
        $\Ex\big[\mathrm{exp}(sNX_N^2)\big]\leq C^N$.
    \end{enumerate}
\end{lemma}
\begin{proof}
    Assume (1), then
    \begin{align*}
        \Ex\big[\mathrm{exp}(sNX_N^2)\big]&=\int_0^\infty \P\big(\mathrm{exp}(sNX_N^2)\geq t\big)\dif t\\
        &\leq 1 + 2\int_0^\infty \P\big(X_N)\geq u\big) sNue^{sNu^2}\dif u\\
        &1 + \int_0^{t_0}2sNue^{sNu^2}\dif u+\int_{t_0}^\infty 2sNue^{-(\alpha-s)Nu^2}\dif u\\
        &=1 + e^{sNt_0^2}+\frac{s}{\alpha - s}
        \leq C^N,
    \end{align*}
    for $C\geq 1+e^{st_0^2}+\tfrac{s}{\alpha-s}$ and $s<\alpha$. 
    Conversely, assume (2), then Chebyshev's inequality yields
    $$
    \P(X_N\geq t)\leq e^{-sNt^2}\Ex\big[\mathrm{exp}(sNX_N^2)\big]
        \leq e^{-sNt^2} C^N
        \leq e^{-\alpha Nt^2},
    $$
    whenever $\alpha < s$ and $t\geq t_0 \sqrt{C/(s-\alpha)}$.
\end{proof}
Recall that the spectrum of an $N\times N$ GUE matrix has the distribution
\begin{equation}\label{eq: gue_spectrum}
\sigma^N(\dif\underline{\lambda})=\frac{1}{Z^N}\prod_{1\leq i<j\leq N}|\lambda_i-\lambda_j|^2
\exp\left\{-\tfrac{N}{2}\sum_{i=1}^N\lambda_i^2\right\}\dif\underline{\lambda},
\end{equation}
where $Z^N$ is the normalisation constant given by
\begin{equation}\label{eq: normalising_constant}
    Z^N=(2\pi)^{\tfrac{N}{2}}N^{-\tfrac{N^2}{2}}\prod_{i=1}^Ni!,
\end{equation}
see \cite[Property 3.1]{Wigner_entropy} with $\beta = 2$.
The following pair of lemmas are key.
\begin{lemma}\label{lem: gue_parition}
    There exists $C>0$ such that, for all $N\geq 2$,
$\frac{Z^{N-1}}{Z^N}\leq e^{CN}$.
\end{lemma}
\begin{proof}
From~\eqref{eq: normalising_constant}, we have
\begin{align*}
    \frac{Z^{N-1}}{Z^N}&=\frac{(N-1)^{-\frac{(N-1)^2}{2}}N^{\tfrac{N^2}{2}}}{\sqrt{2\pi}\Gamma(N+1)}
    =\frac{\exp{\tfrac{1}{2}\big(N^2\log(N)-(N-1)^2\log(N-1)\big)}}{\sqrt{2\pi}N!}.
\end{align*}
It can be checked that for every $N\geq 2$,
$N^2\log(N)-(N-1)^2\log(N-1)\leq 2N\log(N) + 2N$,
and hence
\[
\frac{Z^{N-1}}{Z^N}\leq \frac{\exp{N+N\log N}}{\sqrt{2\pi}\Gamma(N+1)}.
\]
Using the bound from Stirling's formula
$N!\geq \sqrt{2\pi N}e^{N\log N - N}$,
we obtain, 
for suitably large $C>0$,
\[
\frac{Z^{N-1}}{Z^N}\leq \frac{e^{2N}}{2\pi\sqrt{N}}\leq e^{CN}.
\]
\end{proof}
The proof of the next lemma follows the strategy of \cite[Lemma~6.2]{spin_glasses}. In the following we write $\underline{\lambda}^N\sim \sigma^N$ to mean that the vector $\underline{\lambda} \coloneqq (\lambda_1,\dots,\lambda_N)$ is distributed according to $\sigma^N$, as defined in~\eqref{eq: gue_spectrum}.
\begin{lemma}[Tails of the GUE]\label{lem: gue_tail}
    Let $\underline{\lambda}^N\sim \sigma^N$, then there exists $\alpha>0$ and $t_0>0$ such that for all  $t\geq t_0> 0$ and~$N\in \N$, it holds that 
    \begin{equation}\label{eq: gue_tail}
        \P_{\sigma_N}\Big(\max_{i=1,\ldots, N}|\lambda_i|\geq t\Big)\leq e^{-\alpha N t^2}.
    \end{equation}
\end{lemma}
\begin{proof}
    First note that for all $|x|\geq 8$ and $\lambda\in\R$, then
    \[
    |x-\lambda|^2e^{-\tfrac{\lambda^2}{2}}\leq 2(|x|^2+|\lambda|^2)e^{-\tfrac{\lambda^2}{2}}\leq 4|x^2|\leq e^{\tfrac{x^2}{4}},
    \]
    so that for all $|x|\geq t\geq 8$,
    \begin{align*}
        \P_{\sigma_N}\big(|\lambda_1|\geq t\big)&\leq \frac{Z^{N-1}}{Z^N}e^{-\tfrac{1}{4}Nt^2}\int_{|x|\geq t}e^{-\tfrac{x^2}{4}}\int \prod_{i=2}^N|x-\lambda_i|^2e^{-\tfrac{\lambda_i^2}{2}-\tfrac{x^2}{4}}\sigma^{N-1}\big(\dif\underline{\lambda}^{N-1}\big)\dif x\\
        &\leq \frac{Z^{N-1}}{Z^N}e^{-\tfrac{1}{4}Nt^2}\int e^{-\tfrac{x^2}{4}}\int e^{\tfrac{x^2}{4}-\tfrac{x^2}{4}}\sigma^{N-1}\big(\dif\underline{\lambda}^{N-1}\big)\dif x\\
        &\leq \frac{Z^{N-1}}{Z^N}e^{-\tfrac{1}{4}Nt^2}\int e^{-\tfrac{x^2}{4}}\dif x\\
        &\leq e^{CN - \tfrac{1}{4}Nt^2},
    \end{align*}
    for some $C>0$, where the last inequality follows from \cref{lem: gue_parition} and where we have used
    \[
    e^{-N\tfrac{x^2}{2}}=e^{-N\tfrac{x^2}{4}}e^{-N\tfrac{x^2}{4}}\leq e^{-N\tfrac{t^2}{4}}e^{-N\tfrac{x^2}{4}},
    \]
    whenever $|x|\geq t$ to justify the first inequality. Now
    \[
    \P_{\sigma_N}\Big(\max_{i=1,\ldots N}|\lambda_i|\geq t\Big)\leq N\P_{\sigma_N}\big(|\lambda_1|\geq t\big)\leq e^{CN - \tfrac{1}{4}Nt^2},
    \]
    for some new constant $C>0$. Then for any $\alpha<\tfrac{1}{4}$ and $t\geq t_0(\alpha)$,~\eqref{eq: gue_tail} holds for all $N$.
\end{proof}
The proof of \cref{lem: sub_gaussian_tail} now follows that of \cite[Lemma~2.2]{matrix_models}.
\begin{proof}[Proof of \cref{lem: sub_gaussian_tail}]
    Since
    \[
    \lambda_{\max}^N(\mathbf{A})=\max_{1\leq i \leq d}\sup_{\norm{v}=1}\sqrt{\langle v, A_iA_i^{\dag}v\rangle}
    \]
    is a convex function of the entries of the $A_i$'s, then $e^{sN\lambda_{\max}^N(\mathbf{A})^2}$ is also a convex function. Note also that
    \[
    \lambda_{\max}^N(\mathbf{A})^2\leq \big(\lambda_{\max}^N(\mathbf{A}-\mathbf{M})+\lambda_{\max}^N(\mathbf{M})\big)^2\leq 2\lambda_{\max}^N(\mathbf{A}-\mathbf{M})^2+2\lambda_{\max}^N(\mathbf{M})^2,
    \]
    where $\mathbf{M}=\Ex_{\mu_V^N}\big[\mathbf{A}\big]$ is mean of the tuple $\mathbf{A}$. By the proof of \cite[Lemma~2.1]{matrix_models} there is some $C_1>0$ for which $\lambda_{\max}^N(\mathbf{M})<C_1$ for all $N$. Furthermore, by the generalised version of the Brascamp-Lieb inequality
    \[
    \int e^{sN\lambda_{\max}^N(\mathbf{A}-\mathbf{M})^2}\dif \mu_V^N\leq \int e^{sN\lambda_{\max}^N(\mathbf{A})^2}\dif \mu_c^N\leq C_2^N,
    \]
    for all suitably small $s$ and some $C_2>0$, and where $\mu_c^N$ is the law of the GUE ensemble with variance $(cN^{-1})$. The final inequality follows from the sub-Gaussian nature of the tails for the GUE and \cref{lem: sub_gaussian_equiv}. We conclude by observing that for some small $s>0$,
$$
\int e^{sN\lambda_{\max}^N(\mathbf{A})^2}\dif \mu_V^N
\leq e^{2\lambda_{\max}^N(\mathbf{M})^2} \int e^{2sN\lambda_{\max}^N(\mathbf{A}-\mathbf{M})^2}\dif \mu_V^N
\leq e^{C_1} C_2^N,
$$
and applying \cref{lem: sub_gaussian_equiv}.
\end{proof}

\section{Proof of \Cref{thm:sparse_approximation}}\label{proof_sparse_approx}

\begin{definition}[Pauli matrices]
We denote by $\P=\{\sigma_I,\sigma_X,\sigma_Y,\sigma_Z\}$ the set of Pauli matrices, introduced in \cref{def:PS} and write $\sigma\sim \PP$ for a random variable $\sigma$ uniformly distributed over~$\PP$.
\end{definition}
We recall now some basic properties of the Pauli ensemble. Firstly, multiplication of Pauli matrices follows the following rules.
\[
\begin{array}{c|cccc}
  \times & \sigma_I & \sigma_X & \sigma_Y & \sigma_Z \\
  \hline
  \sigma_I & \sigma_I & \sigma_X & \sigma_Y & \sigma_Z \\
  \sigma_X & \sigma_X & \sigma_I & i\sigma_Z & -i\sigma_Y \\
  \sigma_Y & \sigma_Y & -i\sigma_Z & \sigma_I & i\sigma_X \\
  \sigma_Z & \sigma_Z & i\sigma_Y & -i\sigma_X & \sigma_I
\end{array}
\]
We may then consider the group $G=\langle \PP\rangle$ generated by Pauli matrices under matrix multiplication. By the above multiplication rules, it is clear that every $g\in G$ may be uniquely written as
\begin{equation}
    g = \phi(g)m(g),\quad \phi(g)\in \{1,i,-1,-i\}\quad\text{and}\quad m(g)\in \PP.
\end{equation}
We refer to $\phi(g)$ as the phase of $g$. Note the rules of commutation of two elements of $g,h\in G$: $gh=hg$ whenever $m(g)=m(h)$, $m(g)=I$ or $m(h)=I$ and $gh=-hg$ if $m(g)\neq m(h)$ and neither is the identity. Moreover, $m(g)\in \PP$ for any $g\in G$, it follows that $|\tr{g}|\leq 1$ for all $g\in G$.
\begin{lemma}
    Suppose that $X$ and $Y$ are independent random variables with values in $G$ such that $m(X)$ and $m(Y)$ are uniformly distributed on $\PP$, then $m(XY)\sim \PP$.
\end{lemma}
\begin{proof}
    This is clear from the multiplication table and the fact that $m(XY)=m(X)m(Y)$.
\end{proof}
By induction on the above lemma, we obtain the following.
\begin{corollary}\label{cor: uniform_pauli}
    Let $X_1,\ldots,X_n\stackrel{\text{i.i.d.}}{\sim} \PP$, then $m(X_1\cdots X_n)\sim \PP$.
\end{corollary}
We also require the following result.
\begin{lemma}\label{lem: single_cross}
    Let $g\in G$ and $X\sim \PP$, then
    $\Ex[XgX]=\id_{\{m(g)=\sigma_I\}}\phi(g)I$.
\end{lemma}
\begin{proof}
    If $m(g)=I$, then $\Ex[XgX]=\phi(g)\Ex[XX]=\phi(g)I$. If $m(g)\neq I$, then
    \[
    XgX=\begin{cases}
    g, & X=I\text{ or }X=m(g);\\
    -g, & X\neq I\text{ and }X\neq m(g).
    \end{cases}
    \]
    Since both cases have probability $\tfrac{1}{2}$, $\Ex[XgX]=0$.
\end{proof}
We may now combine the above to conclude the following.
\begin{proposition}\label{prop: pauli_products}
    Let $X_1,\ldots,X_p\stackrel{\text{i.i.d.}}{\sim} \PP$ and let $\mathbf{w}$ be a word in $p$ letters of length $2p$ such that each letter $\{1,\ldots,p\}$ appears exactly twice. If pairing equal letters results in a non-crossing pair partition, then $\Ex\big[\tr{X_{w_1}\cdots X_{w_{2p}}}\big]=1$, otherwise
    \begin{equation}
        \Big|\Ex\big[\tr{X_{w_1}\cdots X_{w_{2p}}}\big]\Big|\leq \frac{1}{4}.
    \end{equation}
\end{proposition}
\begin{proof}
    If $\mathbf{w}$ defines a non-crossing pair partition, then any consecutive pairs of letters $X_jX_j$ may be collapsed to the identity. By the non-crossing property, there must be at least one pair, resulting in a shorter word of $2p-2$ letters defining a non-crossing pair partition in $p$-distinct letters. By induction, the entire word may be collapsed to the identity, from which the result follows.
    
    Assume now that $\mathbf{w}$ does not define a non-crossing pair partition. First note that if there are any consecutive letters in $\mathbf{w}$, i.e. the appearance of $X_jX_j$ for some $j\in\{1,\ldots,p\}$, then the pair may be collapsed to the identity. If this collapse results in a new consecutive pair appearing, this may also be collapsed. Since the pairing of equal letters does not result in a non-crossing pair partition, the entire product cannot be collapsed to the identity. As such, we may assume without loss of generality that $\mathbf{w}$ contains no consecutive. We claim that there must be a subword of $\mathbf{w}$ of the form $X_jX_{\mathbf{l}}X_j$ for some word non-empty word $\mathbf{l}$ that contains no repeat letters. This can be checked by induction on the length of the word, with the conclusion that if no such subword exists, then there must be a consecutive pair. We may then write $\mathbf{w}=\mathbf{k}j\mathbf{l}j\mathbf{m}$, where $\mathbf{k},\mathbf{m}$ do not contain $j$. Hence
    \begin{align*}
        \Ex\big[\tr{X_{\mathbf{w}}}\big]&=\tr{\Ex[X_{\mathbf{w}}]}\\
        &=\tr{\Ex[X_{\mathbf{k}}X_jX_{\mathbf{l}}X_jX_{\mathbf{m}}]}\\
        &= \tr{\Ex\big[X_{\mathbf{k}}\Ex[X_jX_{\mathbf{l}}X_j\mid X_i:i\in[p]\setminus{j}]X_{\mathbf{m}}]\big]}\\
        &=\tr{\Ex\big[X_{\mathbf{k}}\Ex[X_jX_{\mathbf{l}}X_j\mid X_{\mathbf{l}}]X_{\mathbf{m}}\big]}\\
        &=\tr{\Ex\big[X_{\mathbf{k}}\phi(X_{\mathbf{l}})\id_{\{m(X_{\mathbf{l}})=I\}}X_{\mathbf{m}}\big]}\\
        &=\Ex\big[\phi(X_{\mathbf{l}})\id_{\{m(X_{\mathbf{l}})=I\}}\tr{X_{\mathbf{k}}X_{\mathbf{m}}}\big].
    \end{align*}
    The fifth equality used \cref{lem: single_cross}. Therefore
$$
\Big|\Ex\big[\tr{X_{\mathbf{w}}}\big]\Big|
\leq\Ex\Big[\id_{\{m(X_{\mathbf{l}})=I\}}\big|\phi(X_{\mathbf{l}})\big|\big|\tr{X_{\mathbf{k}}X_{\mathbf{m}}}\big|\Big]
    \leq \P\big(\{m(X_{\mathbf{l}})=I\}\big)
    =\frac{1}{4},
$$
since $|\phi(g)|=1$ for all $g\in G$ and $|\tr{g}|\leq 1$ for all $g\in G$ and  the final equality uses \cref{cor: uniform_pauli}.
\end{proof}

\begin{lemma}\label{lem: cross_pauli_prod}
     Let $X_1,\ldots,X_p\stackrel{\text{i.i.d.}}{\sim} \PP$ and let $\mathbf{u},\mathbf{v}$ be words in the letters $[p]$ such that
     \begin{enumerate}[label=\arabic*)]
         \item $|\mathbf{u}\mathbf{v}|=2p$;
         \item Each letter appears exactly twice in $\mathbf{u}\mathbf{v}$;
        \item There is at least one letter that appears in both $\mathbf{u}$ and $\mathbf{v}$.
     \end{enumerate}
     Then 
     \[
     \Ex\big[\tr{X_{\mathbf{u}}}\tr{X_{\mathbf{v}}}\big]\leq\frac{1}{4}.
     \]
\end{lemma}

\begin{proof} 
    For any word $\mathbf w$ we have that $\tr{X_{\mathbf w}} = \phi(X_{\mathbf w}) \id_{\{m(X_{\mathbf w})=\sigma_I\}}$. We know there exists a letter $j \in [p]$ that appears once in both $\mathbf u$ and $\mathbf v$, therefore we can decompose the two words as $\mathbf u = \mathbf u_1 j \mathbf u_2$ and $\mathbf v = \mathbf v_1 j \mathbf v_2$, where none of $\mathbf u_1, \mathbf u_2, \mathbf v_1, \mathbf v_2$ contain the letter $j$. Thus, we have
    \begin{equation*}
        |\tr{X_{\mathbf u_1}X_jX_{\mathbf u_2}}\tr{X_{\mathbf v_1}X_jX_{\mathbf v_2}}| = \id_{\{m(X_{\mathbf u_1}X_jX_{\mathbf u_2})=\sigma_I\}} \id_{\{m(X_{\mathbf v_1}X_jX_{\mathbf v_2})=\sigma_I\}}. 
    \end{equation*}
    The two indicator functions are simultaneously equal to $1$ if and only if
    \begin{equation}\label{cond}
        m(X_j) = \alpha^{-1} \quad \text{and} \quad m(X_j) = \beta^{-1}
    \end{equation}
    where $\alpha = m(X_{\mathbf u_1})m(X_{\mathbf u_2})$ and $\beta = m(X_{\mathbf v_1})m(X_{\mathbf v_2})$. \medskip
    
    \noindent
    Now, consider the $\sigma$-algebra $\mathcal{F} = \sigma(X_i : i \neq j)$ generated by all Paulis $X_i$ except $X_j$. Conditioning on $\mathcal{F}$, the random matrix $X_j$ is uniformly distributed on $\PP$. Hence,
    \begin{equation*}
        \mathbb{E}\left[ |\tr{X_{\mathbf u_1}X_jX_{\mathbf u_2}}\tr{X_{\mathbf v_1}X_jX_{\mathbf v_2}}| \mid \mathcal{F}\right] = \mathbb{P}(\text{\eqref{cond} holds} \mid \mathcal{F}) = \begin{cases}
            \frac{1}{4}, & \text{if }\alpha = \beta, \\
            0, & \text{if }\alpha \neq \beta.
        \end{cases}
    \end{equation*}
    In either case, the conditional expectation does not exceed $\frac{1}{4}$. 
    Thus, taking expectation and using the fact that $|\phi(g)|=1$ for all $g\in G$ yields the result.
\end{proof}

The tensor product
$\sigma_{\mathbf{w}}=\sigma_{w_1}\otimes\cdots\otimes\sigma_{w_n}$,
where $\sigma_{w_i}\in \PP$,
is a $2^n\times 2^n$ matrix such that 
\begin{equation}\label{eq: trace_p_string}
\tr{\sigma_{\mathbf{w}}}=\tr{\sigma_{w_1}}\cdots\tr{\sigma_{w_n}}.
\end{equation}
We recall again the definition of the random Pauli string ensemble.
\begin{definition}
Let $N=2^n$. A matrix $A$ belongs to the random Pauli ensemble with~$m$ Pauli strings if
    \[
    A=\frac{1}{\sqrt{m}}\sum_{i=1}^mr^i\sigma^{i},
    \]
    where $r_i\stackrel{\text{i.i.d.}}{\sim}\mathrm{Rad}(\nicefrac{1}{2})$ and $\sigma^{i}=\sigma(1)^i\otimes\cdots\otimes\sigma(n)^{i}$, where $\sigma(j)^i\stackrel{\text{i.i.d.}}{\sim}\PP$.
\end{definition}

\begin{definition} For any alphabet size $m \geq 1$, denoted by $[m]$, and integer $p\geq 1$ define
    \begin{itemize}
        \item $W(m,p):$ number of words of length $2p$ where distinct letters appear an even number of times.
        \item $N(m,p):$ number of words of length $2p$ where distinct letters appear exactly twice. 
    \end{itemize}
\end{definition}

\begin{lemma} 
    For any $m \geq 1$ and $p\geq 1$,
    \begin{equation}\label{eqn:W}
        W(m,p) = \frac{1}{2^m} \sum_{k=0}^m \binom{m}{k} (m-2k)^{2p}
        = \sum_{\substack{k_1 + \cdots + k_m = p\\k_i \geq 0}} \frac{(2p)!}{(2k_1)!\cdots(2k_m)!}. 
    \end{equation}
\end{lemma}
\begin{proof}
For a given word $\mathbf{w}=w_1\cdots w_{2p}$, define the vector $c(\mathbf{w})=(c_1(\mathbf{w}),\ldots,c_m(\mathbf{w}))$, where $c_i(\mathbf{w})$ denotes the number of times letter $i$ appears in $\mathbf{w}$. Write
    \[
    \delta(c_i(\mathbf{w}))=\frac{1+(-1)^{c_i(\mathbf{w})}}{2},
    \]
    for the indicator that the count of letter $i$ is even. Then $\mathbf{w}$ contributes to $W(m,p)$ if and only if
    \[
    1=\prod_{i=1}^m\frac{1+(-1)^{c_i(\mathbf{w})}}{2}=\frac{1}{2^m}\sum_{\epsilon\in\{0,1\}^m}(-1)^{\sum_{i=1}^m\epsilon_ic_i(\mathbf{w})}.
    \]
    Hence, by summing over all words in the alphabet $[m]$ of length $2p$, we have
    \begin{align*}
    W(m,p)&=\sum_{\mathbf{w}}\frac{1}{2^m}\sum_{\epsilon\in\{0,1\}^m}(-1)^{\sum_{i=1}^m\epsilon_ic_i(\mathbf{w})}
    =\frac{1}{2^m}\sum_{\epsilon\in\{0,1\}^m}\sum_{\mathbf{w}}(-1)^{\sum_{i=1}^m\epsilon_ic_i(\mathbf{w})}.
    \end{align*}
    For fixed $\epsilon$, we can compute the inner sum as
    \[
    \sum_{\mathbf{w}}(-1)^{\sum_{i=1}^m\epsilon_ic_i(\mathbf{w})}=\sum_{\mathbf{w}}\prod_{j=1}^{2p}(-1)^{\epsilon_{w_j}}=\left(\sum_{i=1}^m(-1)^{\epsilon_i}\right)^{2p},
    \]
    where $\epsilon_{w_j}=\epsilon_{i}$ whenever $w_j=i$. Noticing that
    \[
    \sum_{i=1}^m(-1)^{\epsilon_i}=m-2t,\quad \text{where }t=\#\{i:\epsilon_i=1\},
    \]
    we can combine the above to get 
    \begin{align*}
        W(m,p)&=\frac{1}{2^m}\sum_{\epsilon\in\{0,1\}^m}\left(\sum_{i=1}^m(-1)^{\epsilon_i}\right)^{2p}
        =\frac{1}{2^m}\sum_{t=0}^m\binom{m}{t}(m-2t)^{2p}.
    \end{align*} 
    To see the second equality, from the conditions defining $W(m,p)$ we have $c_i(\mathbf w) = 2k_i(\mathbf w)$, for some $k_i(\mathbf w) \geq 0$, and $\sum_{i=1}^{m}k_i(\mathbf w) = p$. Using the multinomial theorem we know that the letters of $\mathbf w$ can be arranged in $\frac{(2p)!}{(2k_1(\mathbf w))!\cdots(2k_m(\mathbf w))!}$ different ways. Putting these two observations together yields the result.
\end{proof}

\begin{lemma}
    For any $m \geq 1$ and $p\geq 1$,
    \begin{equation}
        N(m,p) = \binom{m}{p}\frac{(2p)!}{2^p} 
        = \sum_{\substack{k_1 + \cdots + k_m = p\\k_i \in \{0,1\}}} \frac{(2p)!}{(2k_1)!\cdots(2k_m)!}. \label{eqn:N}
    \end{equation}
\end{lemma}
\begin{proof}
    We have to choose which $p$ letters out of the $m$ possible letters will appear, and then arrange those $p$ letters twice each in a word of length $2p$. There are $\binom{m}{p}$ ways of choosing $p$ letters from the alphabet $[m]$. For the chosen $p$ letters, we must place two copies of each letter among the $2p$ positions. For any set of the form $\{w_1, w_1, w_2, w_2, \cdots\}$, the number of distinct orderings is the multinomial count $\frac{(2p)!}{2^p}$, which yields the first equality. The second equality follows from the observation that each multi‐index \((k_1,\ldots,k_m)\) with \(k_i\in\{0,1\}\) and \(\sum_i k_i=p\) has exactly \(p\) of the \(k_i\)’s equal to $1$ and the remaining \(m-p\) equal to $0$.  Hence
    \[
    \frac{(2p)!}{(2k_1)!\,(2k_2)!\,\cdots\,(2k_m)!}
    =\frac{(2p)!}{\underbrace{2!\,\cdots\,2!}_{p\text{ factors}}\;\underbrace{0!\,\cdots\,0!}_{m-p\text{ factors}}}
    =\frac{(2p)!}{(2!)^p}
    =\frac{(2p)!}{2^p}.
    \]
    There are \(\binom{m}{p}\) ways to choose which \(p\) indices satisfy \(k_i=1\), so the total sum is $\binom{m}{p}\,\frac{(2p)!}{2^p}$ as claimed.
\end{proof}

\begin{lemma}\label{lemma:W_bound}
    For any $m \geq 1$ and $p\geq 1$, the inequality 
$\displaystyle W(m,p) \leq 2^p m^p p!$ holds.
\end{lemma}
\begin{proof}
For any integer $k \geq 0$,
$(2k)! = 2^k k!(1 \cdot 3 \cdot \cdots \cdot (2k-1)) \geq 2^k k!$.
    Hence, from~\eqref{eqn:W} we obtain
    \begin{equation*}
       W(m,p) \leq \frac{(2p)!}{2^p} \sum_{\substack{k_1 + \cdots + k_m = p\\k_i \geq 0}} \frac{1}{k_1!\cdots k_m!} = \frac{(2p)!}{2^pp!}m^p
    \end{equation*}
    where we used the fact that
    \begin{equation*}
        \sum_{\substack{k_1 + \cdots + k_m = p\\k_i \geq 0}} \frac{1}{k_1!\cdots k_m!} = \frac{m^p}{p!}.
    \end{equation*}
    Finally the lemma follows since
$(2p)! = 2^pp!(1\cdot 3 \cdots (2p-1)) 
        \leq 2^pp!(2\cdot 4 \cdots 2p) = (2^pp!)^2$.
\end{proof}

\begin{lemma}\label{lemma:W_N_bound}
    For any $m \geq 1$ and $p\geq 1$  we have
    \begin{equation}
        W(m,p) - N(m,p) \leq (4e)^p m^{p-1}p!
    \end{equation}
\end{lemma}
\begin{proof}

Using~\eqref{eqn:W} and~\eqref{eqn:N} we obtain
\begin{equation*}
    W(m,p) - N(m,p) = \sum_{\substack{k_1 + \cdots + k_m = p\\ \text{at least one } k_i \geq 2}} \frac{(2p)!}{(2k_1)!\cdots(2k_m)!}.
\end{equation*}
Using similar arguments as in the proof of \Cref{lemma:W_bound},
\begin{align*}
    W(m,p) - N(m,p) &\leq \frac{(2p)!}{2^p} \sum_{\substack{k_1 + \cdots + k_m = p\\ \text{at least one } k_i \geq 2}} \frac{1}{k_1!\cdots k_m!} \\
    &= \frac{(2p)!}{2^p} \left(\sum_{\substack{k_1 + \cdots + k_m = p\\ k_i \geq 0}} \frac{1}{k_1!\cdots k_m!} - \sum_{\substack{k_1 + \cdots + k_m = p\\ k_i \in \{0,1\}}} \frac{1}{k_1!\cdots k_m!}\right) \\
    &=\frac{(2p)!}{2^p} \left(\frac{m^p}{p!} - \binom{m}{p}\right) \\
    &= \frac{(2p)!}{2^pp!}\left(m^p - \frac{m!}{(m-p)!}\right)\\
    &= \frac{(2p)!}{2^pp!} m^{p-1}\left(m - m\left(1 - \frac{1}{m}\right)\left(1 - \frac{2}{m}\right)\cdots\left(1 - \frac{p-1}{m}\right)\right)
\end{align*}
Now, for any $j = 1,\cdots,p-1$ then $0< \frac{j}{m} < 1$ and $(1 - \alpha) \geq \exp(-\frac{\alpha}{1-\alpha})$ for any $\alpha \in (0,1)$. Hence,
\begin{align*}
    W(m,p) - N(m,p) &\leq \frac{(2p)!}{2^pp!} m^{p-1}\left(m - m \prod_{j=1}^{p-1}\exp(-\frac{j/m}{1-j/m}) \right)\\
    &= \frac{(2p)!}{2^pp!} m^{p-1}\left(m - m \exp(-\frac{1}{m}\sum_{j=1}^{p-1} \frac{j}{1-j/m}) \right) \\
    &\leq \frac{(2p)!}{2^pp!} m^{p-1}\left(m - m \exp(-\frac{p(p-1)}{2m}) \right) \\
    &\leq \frac{(2p)!}{2^pp!} \frac{p(p-1)}{2} m^{p-1} 
    \leq (2p)^p\frac{p^2}{2} m^{p-1}
\end{align*}
where in the penultimate equality we used the simple fact that $1 - e^{-x} \leq x$ for any $x \geq 0$ and set $x = \frac{p(p-1)}{2m}$. The classical Stirling bound $p!\geq (p/e)^p$ easily yields $(2p)^p\frac{p^2}{2} \leq (4e)^p p!$, from which the result follows.
\end{proof}
\begin{proposition}
    For each $n$, suppose that $A_1^n,\ldots,A_d^n$ be i.i.d. matrices from the random Pauli string ensemble of dimension $N=2^n$ with $m_n$ strings each. If $\lim_{n\to\infty}m_n=\infty$, then
    \[
    \lim_{n\to\infty}\Ex\Big[\tr{A_{\bm{w}}^n}\Big]=\tau(X_{\bm{w}}),
    \]
    for all $\bm{w}\in\WW_d$, where $\tau$ is the law of $d$-free semicircular random variables.
\end{proposition}
\begin{proof}
    First note that that
    \[
    \Ex\Big[\tr{A_{\bm{w}}^n}\Big]=m_n^{-\nicefrac{|\bm{w}|}{2}}\sum_{\mathbf{u}}\Ex[r_{\bm{w}}^{\mathbf{u}}]\Ex\big[\tr{\sigma_{\bm{w}}^{\mathbf{u}}}\big],
    \]
    where the sum runs over words $\mathbf{u}$ in $m_n$-letters of length $|\bm{w}|$. Since $\Ex[r_{\bm{w}}^{\mathbf{u}}]$ disappears unless each letter in $\mathbf{u}$ appears an even number of times, we need only consider the case where $|\bm{w}|=2p$ for $p\in\N$. In this case, we still require each letter in $\mathbf{u}$ to appear at least twice so that the total number of distinct letters in $\mathbf{u}$ is at most $p$. We will write $\wt{\mathbf{u}}$ for the number of distinct letters in $\mathbf{u}$. We say that two words $\mathbf{u}$ and $\mathbf{v}$ are equivalent if there is a bijection on $[m_n]$ that takes $\mathbf{u}$ to $\mathbf{v}$. We have the following facts:
    \begin{enumerate}[label=\arabic*)]
        \item The number of words in an equivalence class of a given $m$-word with weight $t$ is exactly
        \[
        m_n(m_n-1)\cdots(m_n-t+1)\leq m_n^t.
        \]
        \item The number of equivalence classes of words with weight $t$ is bounded by $t^{2p}$.
    \end{enumerate}
    By~\eqref{eq: trace_p_string} and the property of traces of the group $G$, then $\big|\Ex\big[\tr{\sigma_{\bm{w}}^{\mathbf{u}}}\big]\big|\leq 1$. Combining these facts yields
    \begin{equation}
        \sum_{\mathbf{u}:\wt{\mathbf{u}}\leq p-1}\Big| \Ex[r_{\bm{w}}^{\mathbf{u}}]\Ex\big[\tr{\sigma_{\bm{w}}^{\mathbf{u}}}\big]\Big|\leq \sum_{t=1}^{p-1}m_n^t  t^{2p}\leq m_n^{p-1} p^{2p+1}.
    \end{equation}
    We turn our attention to words of weight $p$. First, we see that
    \[
    \Ex\big[\tr{\sigma_{\bm{w}}^{\mathbf{u}}}\big]=\Ex\big[\tr{\sigma(1)_{\bm{w}}^\mathbf{u}}\cdots\tr{\sigma(n)_{\bm{w}}^{\mathbf{u}}}\big]=\Ex\big[\tr{\sigma(1)_{\bm{w}}^{\mathbf{u}}}\big]^n,
    \]
    by the i.i.d. construction of the Pauli strings. Consider first case where pairing equal letters in $\mathbf{u}$ defines a crossing pair partition, then by \cref{prop: pauli_products} we see that know that
    \[
    \Big|\Ex\big[\tr{\sigma(1)_{\bm{w}}^{\mathbf{u}}}\big]^{n}\Big|\leq 2^{-2n}.
    \]
    The number of words of length $2p$ with weight $p$, where distinct letters appear exactly twice may be bounded by $m^p \frac{p^{2p}}{p!}$. And so the total contribution of these terms is bounded by
    $m_n^p \frac{p^{2p}}{p!} 2^{-2n}$.
    Finally, consider those words $\mathbf{u}$ where pairing equal letters results in a non-crossing pair partition. Then
    \[
    \Ex[r_{\bm{w}}^{\mathbf{u}}]\Ex\big[\tr{\sigma_{\bm{w}}^{\mathbf{u}}}\big] = \prod_{(a,b)\in\pi_{\mathbf{u}}}\id_{w_a=w_b},
    \]
    where $\pi_\mathbf{u}$ is the pair partition of the letters $\{1,\ldots,2p\}$ defined by the pairing of the letters in $\mathbf{u}$. It is easy to see that two valid words $\mathbf{u}$ and $\mathbf{v}$ are equivalent if and only if the non-crossing pair partitions they each generate are equal. It remains then to count the number of non-crossing pair partitions of $\{1,\ldots,2p\}$ for which
    \[
     \prod_{(a,b)\in\pi_{\mathbf{u}}}\id_{w_a=w_b}=1.
    \]
    Recalling the discussion in \cref{sec: sd_intro}, this is one of the defining properties of $\tau(X_{\bm{w}})$. The number of words in each equivalence class is $m_n(m_n-1)\cdots (m_n-p+1)$. It follows that
    \[
    \lim_{n\to \infty}\Ex\Big[\tr{A_{\bm{w}}^n}\Big]
     = \lim_{n\to\infty}
     \left\{\tau(X_{\bm{w}})\frac{m_n(m_n-1)\cdots (m_n-p+1)}{m^p}+\OO(2^{-2n})+\OO(m_n^{-1})\right\}
     = \tau(X_{\bm{w}}),
    \]
    where the first term is the contribution of the non-crossing pair partitions, the second the crossing pair partitions and the third is the contribution of terms with weight less than $p$.
\end{proof}

We now derive a quantitative version of the previous result.
\begin{proposition}
    For each $n$, suppose that $A_1^n,\ldots,A_d^n$ be i.i.d. matrices from the random Pauli string ensemble of dimension $N=2^n$ with $m_n$ strings each. Then
    \[
    \left|\Ex\Big[\tr{A_{\bm{w}}^n}\Big] - \tau(X_{\bm{w}})\right| \leq (4e)^pp!\left(m_n^{-1} + 2^{-2n}\right),
    \]
    where $\tau$ is the law of $d$-free semicircular random variables and $\bm{w}\in\WW_d$ has length $2p$ for some $p\in\N$.
\end{proposition}
\begin{proof}
    As in the previous proposition, we note that
    \[
    \Ex\Big[\tr{A_{\bm{w}}^n}\Big]=m_n^{-\nicefrac{|\bm{w}|}{2}}\sum_{\mathbf{u}}\Ex[r_{\bm{w}}^{\mathbf{u}}]\Ex\big[\tr{\sigma_{\bm{w}}^{\mathbf{u}}}\big],
    \]
    where the sum runs over words $\mathbf{u}$ in $m$-letters of length $|\bm{w}|$. Again, it suffices to consider words $\mathbf{w}$ of even length $|\mathbf w| = 2p$, with $p \geq 1$. We write $\wt{\mathbf{u}}$ for the number of distinct letters in $\mathbf{u}$. As in the previous proposition, we split this sum into three parts.
    \begin{enumerate}
        \item[(i)] (Words of weight $t<p$) It is easy to see that
        \begin{equation*}
            W(m,p) - N(m,p) = \# \{\mathbf u : \wt{\mathbf u} < p\}
        \end{equation*}
        Thus, using \Cref{lemma:W_N_bound} we obtain directly 
        \begin{equation}
            m_n^{-p}\sum_{\mathbf{u}:\wt{\mathbf{u}}<p}\Big| \Ex[r_{\bm{w}}^{\mathbf{u}}]\Ex\big[\tr{\sigma_{\bm{w}}^{\mathbf{u}}}\big]\Big| = m_n^{-p}(W(m_n,p) - N(m_n,p))
            \leq (4e)^p p! \frac{1}{m_n}\label{i}
        \end{equation}
        \item[(ii)] (Words of weight $t=p$) As before, we need to distinguish two sub-cases:
        \begin{enumerate}
            \item[Case 1:] where pairing equal letters in $\mathbf{u}$ defines a crossing pair partition. By \Cref{prop: pauli_products}
            \begin{equation*}
                \Ex\big[\tr{\sigma_{\bm{w}}^{\mathbf{u}}}\big]=\Ex\big[\tr{\sigma(1)_{\bm{w}}^{\mathbf{u}}}\big]^n \leq 2^{-2n}.
            \end{equation*}
            The count of such words is at most $N(m_n,p) = \binom{m_n}{p}\frac{(2p)!}{2^p} \leq m_n^p\frac{(2p)!}{2^pp!}$. Hence,
            \begin{align}\label{ii}
                m_n^{-p}\sum_{\substack{\mathbf{u}:\wt{\mathbf{u}}=p \\ \text{Case 1}}}\Big| \Ex[r_{\bm{w}}^{\mathbf{u}}]\Ex\big[\tr{\sigma_{\bm{w}}^{\mathbf{u}}}\big]\Big| \leq \frac{(2p)!}{2^pp!}2^{-2n} \leq (4e)^pp!2^{-2n},
            \end{align}
            where we used the Stirling bound $(2p)!/(2^pp!) \leq (4e)^pp!$.
            \item[Case 2:] where pairing equal letters in $\mathbf{u}$ results in a non-crossing pair partition. Using the same counting arguments as in the proof of the previous proposition we have
            \begin{align*}
                m_n^{-p}\sum_{\substack{\mathbf{u}:\wt{\mathbf{u}}=p \\ \text{Case 2}}} \Ex[r_{\bm{w}}^{\mathbf{u}}]\Ex\big[\tr{\sigma_{\bm{w}}^{\mathbf{u}}}\big] &= \tau(X_{\bm{w}})\frac{m_n(m_n-1)\cdots(m_n-p+1)}{m_n^p}
            \end{align*}
            Using the fact that 
            \begin{equation}\label{eq: m_bound}
                \left|\frac{m(m-1)\cdots(m-p+1)}{m^p}-1\right|\leq \frac{p^2}{2m},
            \end{equation}
            yields
            \begin{equation}\label{iii}
                \left|m_n^{-p}\sum_{\substack{\mathbf{u}:\wt{\mathbf{u}}=p \\ \text{Case 2}}} \Ex[r_{\bm{w}}^{\mathbf{u}}]\Ex\big[\tr{\sigma_{\bm{w}}^{\mathbf{u}}}\big| - \tau(X_{\mathbf w})\right| \leq \frac{p^2}{2m_n}|\tau(X_{\bm{w}})| \leq \frac{(4e)^pp!}{m_n},
            \end{equation}
            where we used that fact that $|\tau(X_{\mathbf w})| \leq 4^p$.
        \end{enumerate}
    \end{enumerate}
    Combining the bounds~\eqref{i}, \eqref{ii} and~\eqref{iii} yields
    \[
    \left|\Ex\Big[\tr{A_{\bm{w}}^n}\Big] - \tau(X_{\bm{w}})\right| \leq (4e)^pp!\left(2m_n^{-1} + 2^{-2n}\right),
    \]
\end{proof}

\begin{proposition}
   For each $n$, suppose that $A_1^n,\ldots,A_d^n$ be i.i.d. matrices from the random Pauli string ensemble of dimension $N=2^n$ with $m_n$ strings each. Then
    \[
    \Big|\Ex\Big[\tr{A_{\bm{u}}^n}\tr{A_{\bm{v}}^n}\Big]-\tau(X_{\bm{u}})\tau(X_{\bm{v}})\Big|\leq (4e)^pp!\left(2m_n^{-1} + 2^{-2n}\right),
    \]
    for all words $\bm{u},\bm{v}\in\WW_d$ such that $|\bm{u}|+|\bm{v}|=2p$ for some $p\in\N$, with the left-hand-side being zero otherwise.
\end{proposition}
\begin{proof}
    We first write
    \[
    \Ex\Big[\tr{A_{\bm{u}}^n}\tr{A_{\bm{v}}^n}\Big]=m_n^{-\nicefrac{(|\bm{u}| + |\bm{v}|)}{2}}\sum_{\mathbf{a},\mathbf{b}}\Ex[r_{\bm{u}}^{\mathbf{a}}r_{\bm{v}}^{\mathbf{b}}]\Ex\big[\tr{\sigma_{\bm{u}}^{\mathbf{a}}}\tr{\sigma_{\bm{v}}^{\mathbf{b}}}\big],
    \]
    where the sum runs over words $\mathbf{a}$ and $\mathbf{b}$ in the letters $[m]$ of length $|\bm{u}|$ and $|\bm{v}|$ respectively. Due to the Rademacher random variables, we need only consider the case where $|\bm{u}|+|\bm{v}|=2p$ for some $p\in \N$. Moreover, we require each letter in the word $\mathbf{a}\mathbf{b}$ to appear at least twice for the expectation to be non-zero. As before, we then split the some into three parts.
    \begin{enumerate}[label=(\roman*)]
        \item (Words of weight $t<p$) Using \cref{prop: pauli_products} we directly compute
        \begin{equation}\label{eq: case_1}
        \begin{split}
            m_n^{-p}\sum_{\mathbf{a},\mathbf{b}:\wt{\mathbf{ab}}<p}\Big| \Ex[r_{\bm{u}}^{\mathbf{a}}r_{\bm{v}}^{\mathbf{b}}]\Ex\big[\tr{\sigma_{\bm{u}}^{\mathbf{a}}}\tr{\sigma_{\bm{v}}^{\mathbf{b}}}\big]\Big| &= m_n^{-p}(W(m_n,p) - N(m_n,p)) \\
            &\leq (4e)^p p! \frac{1}{m_n}
        \end{split}
        \end{equation}
        \item (Words of weight $t=p$) We now consider two subcases
        \begin{enumerate}
            \item[Case 1:] where at least one letter appears in both $\mathbf{a}$ and $\mathbf{b}$ or where pairing equal letters in $\mathbf{a}$ and $\mathbf{b}$ results in a crossing. By noting that,
             \begin{align*}
            \Ex\big[\tr{\sigma_{\bm{u}}^{\mathbf{a}}}\tr{\sigma_{\bm{v}}^{\mathbf{b}}}\big]&=\Ex\big[\tr{\sigma(1)_{\bm{u}}^\mathbf{a}}\tr{\sigma(1)_{\bm{v}}^\mathbf{b}}\cdots\tr{\sigma(n)_{\bm{u}}^{\mathbf{a}}}\tr{\sigma(n)_{\bm{v}}^{\mathbf{b}}}\big]\\
            &=\Ex\big[\tr{\sigma(1)_{\bm{u}}^{\mathbf{a}}}\tr{\sigma(1)_{\bm{v}}^{\mathbf{b}}}\big]^n,
            \end{align*}
            we may apply \cref{prop: pauli_products,lem: cross_pauli_prod} to see that
            \begin{equation}\label{eq: case_2_i}
                m_n^{-p}\sum_{\substack{\mathbf{a},\mathbf{b}:\wt{\mathbf{ab}}=p \\ \text{Case 1}}}\Big| \Ex[r_{\bm{u}}^{\mathbf{a}}r_{\bm{v}}^{\mathbf{b}}]\Ex\big[\tr{\sigma_{\bm{u}}^{\mathbf{a}}}\tr{\sigma_{\bm{v}}^{\mathbf{b}}}\big]\Big| \leq \frac{(2p)!}{2^pp!}2^{-2n} \leq (4e)^pp!2^{-2n},
            \end{equation}
            where we have used the fact that the total number of terms in the sum is bounded $N(m_n,p)$.
            \item[Case 2:] where the words $\mathbf{a}$ and $\mathbf{b}$ are distinct and pairing equal letters in $\mathbf{ab}$ results in a non-crossing pair partition. In this case
            \[
            \Ex\big[\tr{\sigma_{\bm{u}}^{\mathbf{a}}}\tr{\sigma_{\bm{v}}^{\mathbf{b}}}\big]=\Ex\big[\tr{\sigma(1)_{\bm{u}}^{\mathbf{a}}}\big]^n\Ex\big[\tr{\sigma(1)_{\bm{v}}^{\mathbf{b}}}\big]^n=1.
            \]
            We then see that
            \[
            \Ex[r_{\bm{u}}^{\mathbf{a}}r_{\bm{v}}^{\mathbf{b}}]\Ex\big[\tr{\sigma_{\bm{u}}^{\mathbf{a}}}\tr{\sigma_{\bm{v}}^{\mathbf{b}}}\big]=\prod_{(a,b)\in\pi_{\mathbf{a}}}\id_{\bm{u}_a=\bm{u}_b}\prod_{(c,d)\in\pi_{\mathbf{b}}}\id_{\bm{v}_c=\bm{v}_d},
            \]
            where $\pi_{\mathbf{a}}$ is the non-crossing pair partition of the letters $\{1,\ldots,|\bm{w}|\}$ defined by the pairing of the letters in $\mathbf{a}$, with $\pi_{\mathbf{b}}$ defined similarly. Words $\mathbf{ab}$ and $\mathbf{cd}$ in Case 2 are then equivalent if and only if $\pi_{\mathbf{a}}=\pi_{\mathbf{c}}$ and $\pi_{\mathbf{b}}=\pi_{\mathbf{d}}$. It follows that the total number of equivalence classes in Case 2 is given by $\tau(X_{\bm{u}})\tau(X_{\bm{w}})$. Using the counting arguments as in the previous proposition, we have that 
            \[
                 m_n^{-p}\sum_{\substack{\mathbf{a},\mathbf{b}:\wt{\mathbf{ab}}=p \\ \text{Case 2}}}\Ex[r_{\bm{u}}^{\mathbf{a}}r_{\bm{v}}^{\mathbf{b}}]\Ex\big[\tr{\sigma_{\bm{u}}^{\mathbf{a}}}\tr{\sigma_{\bm{v}}^{\mathbf{b}}}\big] = \tau(X_{\bm{u}})\tau(X_{\bm{v}})\frac{m_n(m_n-1)\cdots(m_n-p+1)}{m_n^p}.
            \]
            Using the fact
            \[
            \left|\frac{m_n(m_n-1)\cdots(m_n-p+1)}{m_n^p}-1\right|\leq \frac{p^2}{2m},
            \]
            we conclude that
            \begin{equation}\label{eq: case 2_ii}
                \left|m_n^{-p}\sum_{\substack{\mathbf{a},\mathbf{b}:\wt{\mathbf{ab}}=p \\ \text{Case 2}}} \Ex\left[r_{\bm{u}}^{\mathbf{a}}r_{\bm{v}}^{\mathbf{b}}\right]
                \Ex\left[\tr{\sigma_{\bm{u}}^{\mathbf{a}}}\tr{\sigma_{\bm{v}}^{\mathbf{b}}}\right] - \tau(X_{\bm{u}})\tau(X_{\bm{v}})\right|\leq (4e)^pp!\frac{1}{m_n}.
            \end{equation}
        \end{enumerate}
    \end{enumerate}
    Combining~\eqref{eq: case_1}-\eqref{eq: case_2_i}-\eqref{eq: case 2_ii} gives
    \[
    \left|m_n^{-p}\sum_{\mathbf{a},\mathbf{b}} \Ex\left[r_{\bm{u}}^{\mathbf{a}}r_{\bm{v}}^{\mathbf{b}}\right]
    \Ex\left[\tr{\sigma_{\bm{u}}^{\mathbf{a}}}\tr{\sigma_{\bm{v}}^{\mathbf{b}}}\right] - \tau(X_{\bm{u}})\tau(X_{\bm{v}})\right|\leq (4e)^pp!\left(2m_n^{-1}+2^{-2n}\right).
    \]
\end{proof}

As a corollary of all the results, we can now prove \Cref{thm:sparse_approximation}. We state it again for convenience.

\begin{theorem}[\cref{thm:sparse_approximation}]
    For each $n$, suppose that $A_1^n,\ldots,A_d^n$ be i.i.d. matrices from the random Pauli string ensemble of dimension $N=2^n$ with $m_n$ strings each. Let $U^n_\gamma$ be the solution to
    \[
    \dif U^n_\gamma(s,t) = i\sum_{j=1}^d U^n_\gamma(s,t)A_j^n\dif\gamma_t^j,\qquad U^n_\gamma(s,s)=I,
    \]
    and let
    \[
    \pdg{s}{t}=\sum_{\bm{w}\in\WW_d}i^{|\bm{w}|}\tau(X_{\bm{w}})\SS^{\bm{w}}_{s,t}(\gamma).
    \]
    Then there exists $C>0$ depending only on the length of~$\gamma$ and the dimension of the underlying space~$V$ for which
    \begin{equation}
        \Ex\bigg[\Big(\tr{U^n_\gamma(s,t)}-\pdg{s}{t}\Big)^2\bigg]< C\big(m_n^{-1}+2^{-2n}\big).
    \end{equation}
\end{theorem}

\begin{proof}
    We start by expanding
    \begin{align*}
        \Big(\tr{U^n_\gamma(s,t)}-\pdg{s}{t}\Big)^2=&\sum_{\bm{u},\bm{v}\in \WW_d}i^{|\bm{u}|+|\bm{v}|}\Big(\tr{A_{\bm{u}}}\tr{A_{\bm{v}}}-\tau(X_{\bm{u}})\tau(X_{\bm{v}})\\
        &\qquad\qquad+\tau(X_{\bm{u}})\tau(X_{\bm{v}})-\tr{A_{\bm{u}}}\tau(X_{\bm{v}})\\
        &\qquad\qquad-\tau(X_{\bm{u}})\tr{A_{\bm{v}}}+\tau(X_{\bm{u}})\tau(X_{\bm{v}})\Big)\SS_{s,t}^{\bm{u}}(\gamma)\SS_{s,t}^{\bm{v}}(\gamma).
    \end{align*}
    We start by dealing with the first term.
    \begin{align*}
        \sum_{\bm{u},\bm{v}\in\WW_d}\Big|\Ex\big[&\tr{A_{\bm{u}}}\tr{A_{\bm{v}}}\big]-\tau(X_{\bm{u}})\tau(X_{\bm{v}})\Big|\\
        &\leq \sum_{p=0}^\infty\sum_{|\bm{u}|=0}^{2p}\sum_{|\bm{v}|=2p-|\bm{u}|}\Big|\Ex\big[\tr{A_{\bm{u}}}\tr{A_{\bm{v}}}\big]-\tau(X_{\bm{u}})\tau(X_{\bm{v}})\Big|\Big|\SS_{s,t}^{\bm{u}}(\gamma)\SS_{s,t}^{\bm{v}}(\gamma)\Big|\\
        &\leq \left(2m_n^{-1}+2^{-2n}\right)\sum_{p=0}^\infty\sum_{k=0}^{2p}\sum_{\substack{|\bm{u}|=k\\ |\bm{v}|=2p-k}}(4e)^pp!\cdot\Big|\SS_{s,t}^{\bm{u}}(\gamma)\SS_{s,t}^{\bm{v}}(\gamma)\Big|\\
        &\leq \left(2m_n^{-1}+2^{-2n}\right)\sum_{p=0}^\infty\sum_{k=0}^{2p}\frac{(4e)^pp!\cdot d^{2p}\norm{\gamma}_{\text{1-var}}^{2p}}{k!(2p-k)!}\\
        &\leq\left(2m_n^{-1}+2^{-2n}\right)\sum_{p=0}^\infty\frac{(2p+1)(4e)^p\cdot d^{2p}\norm{\gamma}_{\text{1-var}}^{2p}}{p!}\\
        &\leq C_1\big(\norm{\gamma}_{\text{1-var}},d\big)\left(2m_n^{-1}+2^{-2n}\right).
    \end{align*}
    The remaining two terms may each be bounded by
    \begin{align*}
        \sum_{\bm{u},\bm{v}\in\WW_d}\Big|\Ex\big[&\tr{A_{\bm{u}}}\big]-\tau(X_{\bm{u}})\Big|\tau(X_{\bm{v}})\Big|\SS_{s,t}^{\bm{u}}(\gamma)\SS_{s,t}^{\bm{v}}(\gamma)\Big|\\
        &\leq \sum_{\bm{u}\in\WW_d}\Big|\Ex\big[\tr{A_{\bm{u}}}\big]-\tau(X_{\bm{u}})\Big|\Big|\SS_{s,t}^{\bm{u}}(\gamma)\Big|\sum_{\bm{v}\in\WW_d}\tau(X_{\bm{v}})\Big|\SS_{s,t}^{\bm{v}}(\gamma)\Big|\\
        &\leq \left(2m_n^{-1}+2^{-2n}\right)\sum_{p=0}^\infty \frac{(4e)^pd^{2p}\norm{\gamma}_{\text{1-var}}^{2p}}{p!}\sum_{q=0}^\infty \frac{2^qd^{2q}\norm{\gamma}_{\text{1-var}}^{2q}}{(2q)!}\\
        &\leq C_2\big(\norm{\gamma}_{\text{1-var}},d\big)\left(2m_n^{-1}+2^{-2n}\right).
    \end{align*}
    Taking $C=3\max\{C_1,C_2\}$ yields the result.
\end{proof}

\section{Proofs of \cref{thm:quantumalg} and~\cref{thm:classicalalg}}
We now prove the correctness and provide the complexity analysis of the quantum and classical algorithms.

\subsection{Quantum algorithm: Proof of \cref{thm:quantumalg}}

We state \cref{thm:quantumalg} again for convenience.

\begin{theorem}[\cref{thm:quantumalg}]\label{subsec:quantumalgproof}

    The positive integer parameters $m, n, K, M$ may be chosen such that the output of $\mathcal{A}(\{\Delta_{l}^{\nu}\},\epsilon,\delta)$ estimates the unitary path development $\langle \gamma \rangle$ to within additive error $\epsilon$ with probability $1-\delta$:
    \begin{equation}
        \mathbb{P}(|Q_{\Delta}(M,m,n,K) - \langle \gamma \rangle| \ge \epsilon) < \delta.
    \end{equation}
    The circuit requires $\text{log}(1/\epsilon,1/\delta)$ qubits and $\text{poly}(1/\epsilon,1/\delta)$ gates.
\end{theorem}
\begin{proof}
We use the triangle inequality and union bound to find
\begin{equation}\label{eq:uniontrianglebound2}
    \mathbb{P}\big(|Q_{\Delta}(M,m,n,K) - \langle \gamma \rangle| > \epsilon\big) 
    \le \mathbb{P}\left(\left|Q_{\Delta}(M,m,n,K) - \mathbb{E}[\tr{U_{\gamma}(\alpha,n,K)}]\right| > \frac{\epsilon}{3}\right),
\end{equation}
provided
\begin{equation}
    \left|\mathbb{E}[U_{\gamma}(\alpha,n,K)] - \mathbb{E}[U_{\gamma}(\alpha,n)]\right| < \frac{\epsilon}{3}, \quad |\mathbb{E}[U_{\gamma}(\alpha,n)] - \langle \gamma \rangle| < \frac{\epsilon}{3}.
\end{equation}
The first of these inequalities is ensured by \cref{lemma:trotter} provided $K > 3\Delta_{\gamma}^2 m / \epsilon$ and the second is ensured by \cref{thm:sparse_approximation} with $m=n$ provided $n > \max(6C/\epsilon, \log(6C/\epsilon))$. Finally, we ensure that the right-hand side of~\eqref{eq:uniontrianglebound2} is bounded by $\delta$ using Hoeffding's inequality (noting $|\tr{U}|<1$). We thus take $M > \frac{2}{\epsilon^2}\log(2/\delta)$.
\end{proof}

\subsection{Classical Monte Carlo algorithm: Proof of \cref{thm:classicalalg}}\label{subsec:classicalg}
We now prove \cref{thm:classicalalg}. We restate the classical algorithm and the theorem statement for convenience.
\begin{algorithm}[H]
\caption{Classical Sampling $\mathcal{A}^{\text{c}}$}
\begin{algorithmic}
\Require Data: path increments $\{\Delta_l^1,\ldots,\Delta_l^d\}_{l=1}^L$. Parameters: Integers $N>0$, $M>0$ and $K>0$.
\State Draw $d$ independent $N \times N$ GUE matrices $\{A_\nu\}_{\nu=1}^d$
\For{$m = 1$ to $M$}
    \For{$l = 1$ to $L$}
        \State Compute $A_l := \sum_{\nu=1}^d \Delta^{\nu}_lA_\nu$
        \State Compute $U_l^K := (1+A_l/K)^K$ using a truncated approximation dependent on $K>0$
    \EndFor
    \State Compute $Q_m := \mathrm{Tr}(U_1^K U_2^K \cdots U_L^K)$
\EndFor
\State \Return $Q_{\Delta}(M,N,K) := \frac{1}{M} \sum_{m=1}^M Q_m$
\end{algorithmic}
\end{algorithm}

\begin{theorem}
    The output of $\mathcal{A}^{\text{c}}(\{\Delta_i^{\nu}\})$ satisfies
    \begin{equation}
        \mathbb{P}\left(\left|Q_{\Delta}(N,M,K)-\langle \gamma \rangle)\right|>\epsilon \right) < \delta,
    \end{equation}
    for small $\epsilon > 0$ provided that $M > 2 /\epsilon^2 \log(2/ \delta)$ samples are taken, the dimension $N$ of matrices satisfies $N > e^{2 \Delta_{\gamma}}/\epsilon^2$ and the matrix exponential is truncated to $K > \Delta_\gamma e^{2 \Delta_{\gamma}}/\epsilon$.
\end{theorem}
We begin with a series of lemmas.
\begin{lemma}\label{lemma:classicalhoeff}
    The estimator $Q_{\Delta}(N,M,K)$ of $\mathbb{E}[\tr{U(N,K)}]$ satisfies
    \begin{equation}
        \mathbb{P}(|Q_{\Delta}(M,N,K) - \mathbb{E}[\tr{U(N,K)}]| > \epsilon) < \delta,
    \end{equation}
    provided $M > \frac{2}{\epsilon^2} \log(2/\delta)$ by Hoeffding's inequality.
\end{lemma}

\begin{lemma}\label{lemma:classicalalgexp}
    The truncated matrix exponential $U(N,K):= U_1^K \cdots U_L^K$ satisfies
    \begin{equation}
        \mathbb{E}\left[\tr{U(N,K)} - \tr{U(N)}\right] < \epsilon,
    \end{equation}
    provided $K > \frac{\Delta_l e^{2 \Delta_l}}{2 /\epsilon}$.
\end{lemma}
\begin{proof}
    We write $A_l := \sum_{\nu=1}^d \Delta_l^{\nu}A_\nu$. It is then a standard GUE result that
    \begin{equation}
        \|A_l\| \le \max_{\nu}(\Delta_l^{\nu})\left(1 + \frac{\mu}{N^{2/3}} + o(N^{-2/3}) \right),
    \end{equation}
    where $\Delta_l := \max_{\nu}(\Delta_l^{\nu})$ and $\mu$ is the mean of the Tracy-Widom distribution. Now we write $f_K(A) = (1+A/K)^K$. It is straightforward to obtain the bound
    \begin{equation}
        \mathbb{E}\left[\|\exp(A_l) - f_K(A_l)\|\right] \le \frac{\Delta_l}{K}e^{2 \Delta_l}\left(1 + \frac{\mu}{N^{2/3}} + o(N^{-2/3})\right),
    \end{equation}
    for each $l=1,\ldots,L$ this quantity is thus bounded by $\epsilon$ provided $K > \Delta_{\gamma} e^{2 \Delta_{\gamma}} / K$. Finally, we note
    \begin{equation}
        \mathbb{E}\left[\tr{U_1\cdots U_l} - \tr{f_K(A_1) \cdots f_K(A_l)} \right] < \epsilon.
    \end{equation}
\end{proof}

\begin{lemma}\label{lemma:classicalalgdev}
    The convergence of the finite $N$ development to the unitary development satisfies
    \begin{equation}
        \mathbb{E}\left[ |\tr{U(N)} - \langle \gamma \rangle| \right] < \epsilon,
    \end{equation}
    provided $N > e^{2\Delta_{\gamma}} / \epsilon^2$.
\end{lemma}    
We now turn to the proof of \cref{thm:classicalalg}.
\begin{proof}
    We first note that
    \begin{equation}
        \left|Q_\Delta - \langle \gamma \rangle\right| \le \left|Q_{\Delta} - \mathbb{E}[\tr{U(N,K)}]\right| + \left|\mathbb{E}[\tr{U(N)}] - \mathbb{E}[\tr{U(N,K)}]\right| + \left|\mathbb{E}[\tr{U(N)}] - \langle \gamma \rangle\right|.
    \end{equation}
    Thus 
    \begin{equation}
        \mathbb{P}(|Q_\Delta - \langle \gamma \rangle| > \epsilon) \le \mathbb{P}(|Q_\Delta - \mathbb{E}[\tr{U(N,K)}]| > \epsilon) 
    \end{equation}
    provided, from \cref{lemma:classicalalgdev} and~\cref{lemma:classicalalgexp}, $N > e^{2\Delta_{\gamma}} / \epsilon^2$ and $K > \frac{\Delta_\gamma e^{2 \Delta_\gamma}}{2 /\epsilon}$. Finally, we apply Hoeffding's inequality and \cref{lemma:classicalhoeff}, to obtain the result.
\end{proof}

\end{document}